\documentclass[12pt,american,english,aip,jmp,amssymb,amsmath]{revtex4-1}
\usepackage[latin9]{inputenc}
\usepackage{bm}
\usepackage{amsthm}
\usepackage{amsmath}
\usepackage{amssymb}
\usepackage{esint}

\makeatletter
 
 \@ifundefined{textcolor}{}
 {%
   \definecolor{BLACK}{gray}{0}
   \definecolor{WHITE}{gray}{1}
   \definecolor{RED}{rgb}{1,0,0}
   \definecolor{GREEN}{rgb}{0,1,0}
   \definecolor{BLUE}{rgb}{0,0,1}
   \definecolor{CYAN}{cmyk}{1,0,0,0}
   \definecolor{MAGENTA}{cmyk}{0,1,0,0}
   \definecolor{YELLOW}{cmyk}{0,0,1,0}
 }
\theoremstyle{plain}
\newtheorem{thm}{\protect\theoremname}
  \theoremstyle{definition}
  \newtheorem{defn}[thm]{\protect\definitionname}
  \theoremstyle{plain}
  \newtheorem{lem}[thm]{\protect\lemmaname}

\usepackage{dsfont}
\usepackage{chngcntr}

\makeatother

\usepackage{babel}
  \addto\captionsamerican{\renewcommand{\definitionname}{Definition}}
  \addto\captionsamerican{\renewcommand{\lemmaname}{Lemma}}
  \addto\captionsamerican{\renewcommand{\theoremname}{Theorem}}
  \addto\captionsenglish{\renewcommand{\definitionname}{Definition}}
  \addto\captionsenglish{\renewcommand{\lemmaname}{Lemma}}
  \addto\captionsenglish{\renewcommand{\theoremname}{Theorem}}
  \providecommand{\definitionname}{Definition}
  \providecommand{\lemmaname}{Lemma}
\providecommand{\theoremname}{Theorem}

\begin{document}

\title{Functional Wigner representation of BEC quantum dynamics}

\author{B. Opanchuk}

\email{bogdan@opanchuk.net}

\affiliation{Centre for Atom Optics and Ultrafast Spectroscopy, Swinburne University
of Technology, Hawthorn, VIC 3122, Australia}

\author{P. D. Drummond}

\affiliation{Centre for Atom Optics and Ultrafast Spectroscopy, Swinburne University
of Technology, Hawthorn, VIC 3122, Australia}

\date{\today}
\begin{abstract}
We develop a method of simulating the full quantum field dynamics
of multi-mode multi-component Bose-Einstein condensates in a trap.
We use the truncated Wigner representation to obtain a probabilistic
theory that can be sampled. This method produces c-number stochastic
equations which may be solved using conventional stochastic methods.
The technique is valid for large mode occupation numbers. We give
a detailed derivation of methods of functional Wigner representation
appropriate for quantum fields. Our approach describes spatial evolution
of spinor components and properly accounts for nonlinear losses. Such
techniques are applicable to calculating the leading quantum corrections,
including effects like quantum squeezing, entanglement, EPR correlations
and interactions with engineered nonlinear reservoirs. By using a
consistent expansion in the inverse density, we are able to explain
an inconsistency in the nonlinear loss equations found by earlier
authors.
\end{abstract}
\maketitle

\section{Introduction}

The Wigner representation~\cite{Wigner1932,Moyal1947,Dirac1945}
is a convenient and effective method of simulating the dynamics of
bosonic quantum fields~\cite{Drummond1993}, including Bose-Einstein
condensates (BECs)~\cite{Steel1998}. It works best in the limit
of large particle number, where third-order derivative terms in the
Wigner-Moyal time-evolution equation can be truncated, and direct
diagonalization approaches~\cite{Sakmann2009} become computationally
impossible. Large particle number usually implies large numbers of
field modes with significant population, which makes two-mode variational
approaches~\cite{Li2008,Li2009,Sinatra2011} less accurate. This
technique has been applied to a number of quantum dynamics problems
in both quantum optical~\cite{Drummond1993a,Corney2008,Corney2006}
and BEC systems, including fragmentation~\cite{Isella2006,Isella2005,Gross2011},
dissipative atom transport~\cite{Ruostekoski2005}, dynamically unstable
lattice dynamics~\cite{Shrestha2009}, dark solitons~\cite{Martin2010,Martin2010a},
turbulence~\cite{Norrie2005,Norrie2006a}, decoherence~\cite{Egorov2011},
and squeezing~\cite{Opanchuk2012,Opanchuk2012a}. A comparison of
theoretical predictions of quantum fluctuations with experiment has
generally resulted in excellent agreement, provided the large particle
number criterion is met~\cite{Corney2006,Deuar2007}.

The truncated Wigner technique is a numerically robust and useful
method for BEC simulations. Other methods such as the positive-P representation~\cite{Drummond1980}
are known to work better~\cite{Deuar2007} when the truncation approximation
breaks down, and there are a number of studies of applicability that
compare the truncated Wigner method with the exact positive-P method~\cite{Chaturvedi2002,Dechoum2004}
or, where feasible, Bloch-basis approaches. The typical result found
is that the truncated Wigner method gives correct results out to a
characteristic break time. At this stage, the accumulated errors can
lead to large discrepancies in quantum correlations. The method is
weakest when dealing with nonlinear quantum tunneling~\cite{Drummond1989,Kinsler1991a},
which depends on both long time dynamics and quantum correlations.
Within its domain of applicability the technique is remarkably accurate
and stable. The overall picture of how this method is related to other
techniques for quantum dynamics has been recently reviewed\cite{He2012}.

The phase-space treatment of multimode problems can be simplified
by working with functional mappings to field operators rather than
to single-mode operators. This approach was initially introduced by
Graham~\cite{Graham1970,Graham1970a}. Later it was used in a number
of works~\cite{Steel1998,Gardiner2003,Isella2006,Norrie2006a,Blakie2008}
without formally defining the corresponding transformations or accompanying
theorems; a more detailed description was given by Polkovnikov~\cite{Polkovnikov2010}.
In order to calculate the approximate evolution of the Wigner function
of a system numerically, one has to truncate third-order derivative
terms~\cite{Drummond1993,Steel1998,Sinatra2002}, and project out
modes with low occupation numbers. This further complicates the formal
description of the method. Recent developments in ultra-cold atomic
physics mean that processes like nonlinear damping, not considered
in detail previously, have also become important. Accordingly, much
of the mathematical derivation of these techniques is not readily
available.

In this paper we present a formal description of the application of
the resulting truncated Wigner representation to simulating the multi-mode
dynamics of Bose-Einstein condensates (BECs). We successively reduce
the problem in its initial form, the master equation for bosonic field
operators, to a system of stochastic differential equations, which
have significantly lower computational complexity. While there is
a price for making the truncation approximation, we emphasize that
this is a systematic expansion in a small parameter, $1/N$, where
$N$ is the particle number. Such expansions are also relevant to
stochastic diagram techniques~\cite{Chaturvedi1999}, which can be
used to formally calculate order-by-order behaviour in such equations.
Although not treated here in detail, our identities can be applied
to parametric interactions, where the truncation approximation has
also been applied to EPR and entanglement problems and compared to
more rigorous positive-P simulation methods~\cite{Chaturvedi2002,Dechoum2004}.

The purpose of this paper is to provide a more rigorous proof, within
the functional analysis formalism, of several identities that are
used for these derivations. We focus especially on the problem of
nonlinear damping. This is a dominant relaxation mechanism in BEC
systems, and is often ignored or (incorrectly) approximated using
linear loss terms. We derive the correct Fokker-Planck drift and noise
terms for general multicomponent damping using the $1/N$ expansion,
which transforms to an expansion in the inverse particle density for
quantum fields. Even in the single-component case, the drift term
has both a leading (classical) term and a quantum noise correction
to the damping. This is needed to predict the loss behaviour correctly,
and is important in high-accuracy simulations. Such corrections \textemdash{}
both in the drift and noise \textemdash{} are relevant to topics like
EPR correlations, entanglement and quantum squeezing in the presence
of nonlinear reservoirs, a topic of increasing importance in areas
ranging from quantum optics and BEC physics to nanomechanical oscillators~\cite{Chaturvedi1977,Reid1986a,Rabl2004}.

We derive the resulting stochastic differential equations from the
functional Fokker-Planck equations, and show when the corresponding
truncation approximations are applicable. The final equations can
be treated using standard computational techniques for solving ordinary
and partial stochastic differential equations~\cite{Drummond1990,Werner1997,Wilkie2005}.
There are code generator packages and public domain websites with
code available for this purpose~\cite{Collecutt2001,Dennis2013}.

\section{Quantum fields and dynamics}

In this paper we consider a $C$-component Bose gas in $D$ effective
dimensions. The Hamiltonian for this system is expressed in terms
of bosonic field creation and annihilation operators $\hat{\Psi}_{j}^{\dagger}(\boldsymbol{x})$
and $\hat{\Psi}_{j}(\boldsymbol{x})$, $j=1\ldots C$, which obey
standard bosonic commutation relations
\begin{equation}
[\hat{\Psi}_{j},\hat{\Psi}_{k}^{\prime\dagger}]=\delta_{jk}\delta(\boldsymbol{x}^{\prime}-\boldsymbol{x}).\label{eqn:master-eqn:commutators}
\end{equation}
Here $\boldsymbol{x}\in\mathbb{R}^{D}$ is a $D$-dimensional coordinate
vector, we define $\hat{\Psi}_{j}\equiv\hat{\Psi}_{j}(\boldsymbol{x})$
and $\hat{\Psi}_{k}^{\prime}\equiv\hat{\Psi}_{k}(\boldsymbol{x}^{\prime})$
for brevity (the same abbreviation will be used for all functions
of coordinates), and $\delta(\boldsymbol{x}^{\prime}-\boldsymbol{x})$
is a D-dimensional Dirac delta function.

\subsection{Quantized Hamiltonian}

The second-quantized Hamiltonian for the system, integrated with a
$D-$dimensional volume measure $d\boldsymbol{x}$, is
\begin{equation}
\hat{H}=\int d\boldsymbol{x}\left\{ \hat{\Psi}_{j}^{\dagger}K_{jk}\hat{\Psi}_{k}+\frac{1}{2}\int d\boldsymbol{x}^{\prime}\hat{\Psi}_{j}^{\dagger}\hat{\Psi}_{k}^{\prime\dagger}U_{jk}(\boldsymbol{x}^{\prime}-\boldsymbol{x})\hat{\Psi}_{j}^{\prime}\hat{\Psi}_{k}\right\} ,\label{eqn:master-eqn:hamiltonian}
\end{equation}
where $U_{jk}$ is the two-body scattering potential, and the single-particle
Hamiltonian $K_{jk}$ is
\begin{equation}
K_{jk}=\left(-\frac{\hbar^{2}}{2m}\nabla^{2}+\hbar\omega_{j}+V_{j}(\boldsymbol{x})\right)\delta_{jk}+\hbar\Omega_{jk}(t).\label{eqn:master-eqn:single-particle}
\end{equation}
Here $m$ is the atomic mass, $V_{j}$ is the external trapping potential
for spin $j$, $\hbar\omega_{j}$ is the internal energy of spin $j$,
and $\Omega_{jk}$ represents a time-dependent coupling that is used
to rotate one spin projection into another.

If we impose a momentum cutoff $k_{\mathrm{c}}$ and only take into
account low-energy modes, the non-local scattering potential $U_{jk}(\boldsymbol{x}^{\prime}-\boldsymbol{x})$
can be replaced by the contact potential $U_{jk}\delta(\boldsymbol{x}^{\prime}-\boldsymbol{x})$~\cite{Morgan2000},
giving the effective Hamiltonian
\begin{equation}
\hat{H}=\int d\boldsymbol{x}\left\{ \tilde{\Psi}_{j}^{\dagger}K_{jk}\tilde{\Psi}_{k}+\frac{U_{jk}}{2}\tilde{\Psi}_{j}^{\dagger}\tilde{\Psi}_{k}^{\dagger}\tilde{\Psi}_{j}\tilde{\Psi}_{k}\right\} ,\label{eqn:master-eqn:effective-H}
\end{equation}
where $\tilde{\Psi}_{j}^{\dagger}$ and $\tilde{\Psi}_{j}$ are field
operators in the new restricted basis of low-energy modes, which is
described in detail in the next section. For $s$-wave scattering
in three dimensions the coefficient is $U_{jk}=4\pi\hbar^{2}a_{jk}/m$,
where $a_{jk}$ is the scattering length. In general, the coefficient
must be renormalized depending on the momentum cutoff~\cite{Sinatra2002,Kokkelmans2002},
but the change is small if $dx_{i}\gg a_{jk}$, where $dx_{i}$ is
the grid step in dimension $i$.

\subsection{Master Equation}

The time-evolution of the quantum density matrix $\hat{\rho}$ with
particle losses included can be written as a Markovian master equation~\cite{Jack2002}
for the system:
\begin{equation}
\frac{d\hat{\rho}}{dt}=-\frac{i}{\hbar}\left[\hat{H},\hat{\rho}\right]+\sum_{\boldsymbol{l}}\kappa_{\boldsymbol{l}}\int d\boldsymbol{x}\mathcal{L}_{\boldsymbol{l}}\left[\hat{\rho}\right],\label{eqn:master-eqn:master-eqn}
\end{equation}
where $\boldsymbol{l}=(l_{1},l_{2},\ldots,l_{C})$ is a tuple indicating
the number of atoms from each component involved in the inelastic
interaction that causes the relevant loss. Here we have introduced
local Liouville loss terms, that describe $n$-body collisional losses
in the Markovian approximation:
\begin{equation}
\mathcal{L}_{\boldsymbol{l}}\left[\hat{\rho}\right]=2\hat{O}_{\boldsymbol{l}}\hat{\rho}\hat{O}_{\boldsymbol{l}}^{\dagger}-\hat{O}_{\boldsymbol{l}}^{\dagger}\hat{O}_{\boldsymbol{l}}\hat{\rho}-\hat{\rho}\hat{O}_{\boldsymbol{l}}^{\dagger}\hat{O}_{\boldsymbol{l}}.\label{eqn:master-eqn:loss-term}
\end{equation}

The reservoir coupling operators $\hat{O}_{\boldsymbol{l}}$ are products
of local field annihilation operators:
\begin{equation}
\hat{O}_{\boldsymbol{l}}\equiv\hat{O}_{\boldsymbol{l}}(\tilde{\bm{\Psi}})=\prod_{j=1}^{C}\tilde{\Psi}_{j}^{l_{j}}(\boldsymbol{x}),
\end{equation}
describing local $n$-body collision losses where $n=\left(\sum_{j=1}^{C}l_{j}\right)$$ $.
There is an implicit physical assumption that $l_{1},\ldots l_{C}$,
particles of internal state quantum number $j=1,\ldots C$ all collide
simultaneously within the volume corresponding to the inverse momentum
cutoff, and are removed from the Bose gas.

This is a minimal approach to the complicated issue of particle loss,
since it assumes that the reservoir of ``lost'' particles does not
interact with the original Bose gas. The accuracy of this approach
depends on such issues as the trapping mechanism. Since, for massive
particles, the particle number is conserved in the non-relativistic
limit, ``lost'' particles are simply in a different quantum state.
The assumption that these particles don't interact with the original
Bose gas is only valid if the trap is state-selective or the collision
is highly exothermic, such that the resulting particles are able to
move rapidly away.

It is also possible to treat non-Markovian reservoirs within this
formalism, by extending the Hamiltonian to include the detailed loss
dynamics, but this is not treated in detail in the present paper.

\subsection{Field operators and restricted basis\label{sec:func-operators}}

It is always necessary in interacting quantum field theory to define
a renormalization together with a momentum cutoff. In the case of
the truncated Wigner method, this is more important, as the validity
of the truncation approximation may depend on it (see Section~\ref{sec:Wigner-truncation}
for more details). We therefore wish to treat this momentum cutoff
procedure more carefully, as it has a direct effect on the number
of modes, and hence on the validity of the truncation. For each component
$j$ we define an orthonormal basis consisting of $\phi_{j,\boldsymbol{n}}(\boldsymbol{x})$,
where $\boldsymbol{n}\in\mathbb{B}_{j}$ is a mode identifier. The
orthonormality and completeness conditions for basis functions are,
respectively,
\begin{eqnarray}
\int\limits _{A}\phi_{j,\boldsymbol{n}}^{*}\phi_{j,\boldsymbol{m}}d\boldsymbol{x} & = & \delta_{\boldsymbol{n}\boldsymbol{m}},\nonumber \\
\sum_{\boldsymbol{n}}\phi_{j,\boldsymbol{n}}^{*}\phi_{j,\boldsymbol{n}}^{\prime} & = & \delta(\boldsymbol{x}^{\prime}-\boldsymbol{x}),
\end{eqnarray}
where the exact nature of integration area $A$ depends on the basis
set. For example, $A$ is the whole space for harmonic oscillator
modes, or a box for plane waves. We assume that the integration $\int d\boldsymbol{x}$
is always performed over $A$.

Standard bosonic field operators from~(\ref{eqn:master-eqn:commutators})
can be decomposed as
\begin{equation}
\hat{\Psi}_{j}=\sum_{\boldsymbol{n}\in\mathbb{B}_{j}}\phi_{j,\boldsymbol{n}}\hat{a}_{j,\boldsymbol{n}},
\end{equation}
where single mode operators $\hat{a}_{j,\boldsymbol{n}}$ obey bosonic
commutation relations, the pair $j,\boldsymbol{n}$ serving as a mode
identifier. The cutoff mentioned in the previous section will result
in operating with some fixed subset of each component's basis. Let
$\mathbb{M}_{j}\subseteq\mathbb{B}_{j}$ be these subsets. Restricted
field operators contain only modes from the subset $\mathbb{M}_{j}$:
\begin{equation}
\tilde{\Psi}_{j}=\sum_{\boldsymbol{n}\in\mathbb{M}_{j}}\phi_{j,\boldsymbol{n}}\hat{a}_{j,\boldsymbol{n}}.\label{eqn:func-operators:restricted-psi}
\end{equation}
 Formally, these field operators have the functional type $\tilde{\Psi}_{j}\in\mathbb{FH}_{\mathbb{M}_{j}}\equiv(\mathbb{R}^{D}\rightarrow\mathbb{H}_{\mathbb{M}_{j}})$,
where $\mathbb{H}_{\mathbb{M}_{j}}$ is the Hilbert space of the restricted
subset of modes.

Because of the restricted nature of the operator, commutation relations~(\ref{eqn:master-eqn:commutators})
no longer apply. The following ones should be used instead:

\begin{equation}
\left[\tilde{\Psi}_{j},\tilde{\Psi}_{k}^{\prime}\right]=\left[\tilde{\Psi}_{j}^{\dagger},\tilde{\Psi}_{k}^{\prime\dagger}\right]=0,\qquad\left[\tilde{\Psi}_{j},\tilde{\Psi}_{k}^{\prime\dagger}\right]=\delta_{jk}\delta_{\mathbb{M}_{j}}(\boldsymbol{x}^{\prime},\boldsymbol{x}),\label{eqn:func-operators:restricted-commutators}
\end{equation}
where $\delta_{\mathbb{M}}$ is a restricted delta-function. The definition
is given in the Appendix, in Definition~\ref{def:func-calculus:restricted-delta}.

\section{Functional Wigner representation}

In this section we introduce and obtain properties of the functional
Wigner representation. This will use a number of definitions and results
from functional analysis. The relevant mathematical material that
is used in this section is defined and its properties derived in Appendix
\ref{sec:app-functional-calculus}.

\subsection{Single-mode Wigner transformation}

As a starting point, we recall that the single-mode Wigner transformation
of the operator $\hat{A}$ is defined as
\begin{equation}
\mathcal{W}_{\mathrm{sm}}[\hat{A}]=\frac{1}{\pi^{2}}\int d^{2}\lambda\exp(-\lambda\alpha^{*}+\lambda^{*}\alpha)\mathrm{Tr}\left\{ \hat{A}\hat{D}(\lambda,\lambda^{*})\right\} ,\label{eq:wigner-transformation-sm}
\end{equation}
where the displacement operator $\hat{D}(\lambda,\lambda^{*})=\exp(\lambda\hat{a}^{\dagger}-\lambda^{*}\hat{a})$
was first introduced by Weyl~\cite{Weyl1950}. The detailed description
of the single-mode Wigner function $W(\alpha,\alpha^{*})\equiv\mathcal{W}_{\mathrm{sm}}[\hat{\rho}]$,
analogous to the one provided here, was given by Moyal~\cite{Moyal1947}
or, using a notation close to the one in this paper, by later authors~\cite{Cahill1969,Hillery1984,Gardiner2004}.
In this subsection we will briefly outline these results.

The first theorem provides a way to transform any master equation
written in terms of creation and annihilation operators to a partial
differential equation for the Wigner function. In the case of the
Wigner function, we take $\hat{A}=\hat{\rho}$
\begin{thm}
For any Hilbert-Schmidt operator $\hat{A}$
\begin{eqnarray}
\mathcal{W}_{\mathrm{sm}}[\hat{a}\hat{A}] & =\left(\alpha+\frac{1}{2}\frac{\partial}{\partial\alpha^{*}}\right)\mathcal{W}_{\mathrm{sm}}[\hat{A}],\qquad\mathcal{W}_{\mathrm{sm}}[\hat{a}^{\dagger}\hat{A}]=\left(\alpha^{*}-\frac{1}{2}\frac{\partial}{\partial\alpha}\right)\mathcal{W}_{\mathrm{sm}}[\hat{A}],\nonumber \\
\mathcal{W}_{\mathrm{sm}}[\hat{A}\hat{a}] & =\left(\alpha-\frac{1}{2}\frac{\partial}{\partial\alpha^{*}}\right)\mathcal{W}_{\mathrm{sm}}[\hat{A}],\qquad\mathcal{W}_{\mathrm{sm}}[\hat{A}\hat{a}^{\dagger}]=\left(\alpha^{*}+\frac{1}{2}\frac{\partial}{\partial\alpha}\right)\mathcal{W}_{\mathrm{sm}}[\hat{A}].
\end{eqnarray}

\end{thm}
This is paired with a second theorem, which helps extract observables
(again, expressed in terms of $\hat{a}^{\dagger}$ and $\hat{a}$)
from the Wigner function.
\begin{thm}
For any non-negative integers $r_{j}$, $s_{j}$
\begin{equation}
\langle\left\{ \hat{a}^{r}(\hat{a}^{\dagger})^{s}\right\} _{\mathrm{sym}}\rangle=\int d^{2}\alpha\left(\alpha^{r}(\alpha^{*})^{s}\right)W(\alpha,\alpha^{*}),
\end{equation}

where $\left\{ \right\} _{\mathrm{sym}}$ stands for a symmetrically
ordered product of operators.
\end{thm}

\subsection{Definitions of functional operators}

In order to specify domains and ranges of discussed functions, functionals
and transformations formally, we will employ a special notation. In
general, $F\in A\rightarrow B\rightarrow C$ will denote a function
$F$ that depends on two values of types $A$ and $B$, and has a
value of type $C$. Note this expression at the same time describes
a function that depends on a value of type $A$, and returns a function
with a type $B\rightarrow C$. The types can be nested, for example
$(A\rightarrow B)\rightarrow(C\rightarrow D)$ denotes a function
to function mapping.

We introduce complex functions $\Lambda\left(\mathbf{x}\right)$,
which play the role of the characteristic c-number $\lambda$ in the
single-mode case. The important part of the definition is the functional
analogue of the displacement operator.
\begin{defn}
Functional displacement operator $\hat{D}_{j}\in\mathbb{F}_{\mathbb{M}_{j}}\rightarrow\mathbb{H}_{\mathbb{M}_{j}}$
\[
\hat{D}_{j}[\Lambda,\Lambda^{*}]=\exp\int d\boldsymbol{x}\left(\Lambda\tilde{\Psi}_{j}^{\dagger}-\Lambda^{*}\tilde{\Psi}_{j}\right),
\]

where $\mathbb{F}_{\mathbb{M}_{j}}$, by analogy with $\mathbb{H}_{\mathbb{M}_{j}}$,
is a space of functions that can be decomposed in terms of mode functions
from the subset $\mathbb{M}_{j}$: $\Lambda\equiv\sum_{\boldsymbol{n}\in\mathbb{M}_{j}}\phi_{j,\boldsymbol{n}}\lambda_{j,\boldsymbol{n}}$.
\end{defn}
It is also convenient to define the displacement functional as:
\begin{defn}
Displacement functional $D\in\mathbb{F}_{\mathbb{M}_{j}}\rightarrow\mathbb{F}_{\mathbb{M}_{j}}\rightarrow\mathbb{C}$
\[
D[\Lambda,\Lambda^{*},\Psi,\Psi^{*}]=\exp\int d\boldsymbol{x}\left(-\Lambda\Psi^{*}+\Lambda^{*}\Psi\right).
\]
It can be shown that the functional displacement operator has properties
similar to its single-mode equivalent.\end{defn}
\begin{lem}
\label{lmm:func-wigner:displacement-derivatives}
\begin{eqnarray}
 &  & \frac{\delta}{\delta\Lambda^{\prime}}\hat{D}_{j}[\Lambda,\Lambda^{*}]=\hat{D}_{j}[\Lambda,\Lambda^{*}](\tilde{\Psi}_{j}^{\prime\dagger}+\frac{1}{2}\Lambda^{\prime*})=(\tilde{\Psi}_{j}^{\prime\dagger}-\frac{1}{2}\Lambda^{\prime*})\hat{D}_{j}[\Lambda,\Lambda^{*}],\nonumber \\
 &  & -\frac{\delta}{\delta\Lambda^{\prime*}}\hat{D}_{j}[\Lambda,\Lambda^{*}]=\hat{D}_{j}(\Lambda,\Lambda^{*})(\tilde{\Psi}_{j}^{\prime}+\frac{1}{2}\Lambda^{\prime})=(\tilde{\Psi}_{j}^{\prime}-\frac{1}{2}\Lambda^{\prime})\hat{D}_{j}[\Lambda,\Lambda^{*}].
\end{eqnarray}
\end{lem}
\begin{proof}
Proved using the Baker-Hausdorff theorem and evaluating integrals.
\end{proof}

\subsection{Functional Wigner transformation}

In this subsection we will extend the single-mode definition~\eqref{eq:wigner-transformation-sm}
to the multimode case, using a functional notation.
\begin{defn}
\label{def:func-wigner:w-transformation} A multi-component functional
Wigner transformation $\mathcal{W}\in\left(\mathbb{R}^{D}\rightarrow\prod_{j=1}^{C}\mathbb{H}_{\mathbb{M}_{j}}\right)\rightarrow\prod_{j=1}^{C}\mathbb{F}_{\mathbb{M}_{j}}\rightarrow\mathbb{C}$
is defined as
\[
\mathcal{W}[\hat{A}]=\frac{1}{\pi^{2\sum|\mathbb{M}_{j}|}}\int\delta^{2}\boldsymbol{\Lambda}\left(\prod_{j=1}^{C}D[\Lambda_{j},\Lambda_{j}^{*},\Psi_{j},\Psi_{j}^{*}]\right)\mathrm{Tr}\left\{ \hat{A}\prod_{j=1}^{C}\hat{D}_{j}[\Lambda_{j},\Lambda_{j}^{*}]\right\} ,
\]
where $\Lambda_{j}\in\mathbb{F}_{\mathbb{M}_{j}}$, and $\int\delta^{2}\boldsymbol{\Lambda}\equiv\int\delta^{2}\Lambda_{1}\ldots\delta^{2}\Lambda_{C}$.
The notation $\vert\mathbb{M}_{j}\vert$ stands for the number of
elements in the set $\mathbb{M}_{j}$, so $\sum|\mathbb{M}_{j}|$
is the total number of modes in all restricted mode subsets. This
transforms a coordinate-dependent operator $\hat{A}$ on a restricted
subset of a Hilbert space to a functional $(\mathcal{W}[\hat{A}])[\boldsymbol{\Psi},\boldsymbol{\Psi}^{*}]$.
\end{defn}
Next we introduce the Wigner functional, which is a special case of
Wigner transformation.
\begin{defn}
\label{def:func-wigner:w-functional} The Wigner functional $W\in\prod_{j=1}^{C}\mathbb{F}_{\mathbb{M}_{j}}\rightarrow\mathbb{R}$
is
\[
W[\boldsymbol{\Psi},\boldsymbol{\Psi}^{*}]\equiv\mathcal{W}[\hat{\rho}]=\frac{1}{\pi^{2M}}\int\delta^{2}\boldsymbol{\Lambda}\left(\prod_{j=1}^{C}D[\Lambda_{j},\Lambda_{j}^{*},\Psi_{j},\Psi_{j}^{*}]\right)\chi_{W},
\]
where $\chi_{W}[\boldsymbol{\Lambda},\boldsymbol{\Lambda}^{*}]$ is
the characteristic functional
\begin{equation}
\chi_{W}[\boldsymbol{\Lambda},\boldsymbol{\Lambda}^{*}]=\mathrm{Tr}\left\{ \hat{\rho}\prod_{j=1}^{C}\hat{D}_{j}[\Lambda_{j},\Lambda_{j}^{*}]\right\} .
\end{equation}

\end{defn}
The Wigner functional has two important properties analogous to the
single-mode case. The first one is used to successively transform
operator products.
\begin{thm}
\label{thm:func-wigner:correspondences} For any Hilbert-Schmidt operator
$\hat{A}$, if $\mathcal{W}[\hat{A}]\equiv(\mathcal{W}[\hat{A}])[\boldsymbol{\Psi},\boldsymbol{\Psi}^{*}]$,
then
\begin{eqnarray}
\mathcal{W}[\tilde{\Psi}_{j}\hat{A}] & =\left(\Psi_{j}+\frac{1}{2}\frac{\delta}{\delta\Psi_{j}^{*}}\right)\mathcal{W}[\hat{A}],\qquad\mathcal{W}[\tilde{\Psi}_{j}^{\dagger}\hat{A}]=\left(\Psi_{j}^{*}-\frac{1}{2}\frac{\delta}{\delta\Psi_{j}}\right)\mathcal{W}[\hat{A}],\nonumber \\
\mathcal{W}[\hat{A}\tilde{\Psi}_{j}] & =\left(\Psi_{j}-\frac{1}{2}\frac{\delta}{\delta\Psi_{j}^{*}}\right)\mathcal{W}[\hat{A}],\qquad\mathcal{W}[\hat{A}\tilde{\Psi}_{j}^{\dagger}]=\left(\Psi_{j}^{*}+\frac{1}{2}\frac{\delta}{\partial\Psi_{j}}\right)\mathcal{W}[\hat{A}].
\end{eqnarray}
\end{thm}
\begin{proof}
The proof uses Lemma~\ref{lmm:func-wigner:displacement-derivatives}
given above to transform the $\hat{A}\prod_{j}\hat{D}_{j}$ product
inside the trace, together with Lemma~\ref{lmm:func-calculus:zero-integrals}
from the Appendix to integrate by parts (because of the restriction
on $\hat{A}$, the traces will be square-integrable~\cite{Cahill1969}),
effectively moving the differentials to their intended places.
\end{proof}
The second property complements the first one, providing a way to
obtain expectations of operator products given the Wigner function.
Again, it requires a supplementary lemma.
\begin{lem}
\label{lmm:func-wigner:moments-from-chi} For any non-negative integer
$r$ and $s$:
\begin{equation}
\langle\left\{ (\tilde{\Psi}_{j}^{\prime})^{r}(\tilde{\Psi}_{j}^{\prime\dagger})^{s}\right\} _{\mathrm{sym}}\rangle=\left.\left(\frac{\delta}{\delta\Lambda_{j}^{\prime}}\right)^{s}\left(-\frac{\delta}{\delta\Lambda_{j}^{\prime*}}\right)^{r}\chi_{W}[\boldsymbol{\Lambda},\boldsymbol{\Lambda}^{*}]\right|_{\boldsymbol{\Lambda}\equiv0}.
\end{equation}
\end{lem}
\begin{proof}
The factor corresponding to the $j$-th component in the displacement
operator can be expanded as
\begin{equation}
\exp\int d\boldsymbol{x}(\Lambda_{j}\tilde{\Psi}_{j}^{\dagger}-\Lambda_{j}^{*}\tilde{\Psi}_{j})=\sum_{r,s}\frac{1}{r!s!}\left\{ \left(\int d\boldsymbol{x}\Lambda_{j}\tilde{\Psi}_{j}^{\dagger}\right)^{r}\left(-\int d\boldsymbol{x}\Lambda_{j}^{*}\tilde{\Psi}_{j}\right)^{s}\right\} _{\mathrm{sym}}.
\end{equation}
We can swap functional derivatives with both integration and multiplication
by an independent function, so:
\begin{equation}
\frac{\delta}{\delta\Lambda_{j}^{\prime}}\left(\int d\boldsymbol{x}\Lambda_{j}\tilde{\Psi}_{j}^{\dagger}\right)^{r}=r\tilde{\Psi}_{j}^{\prime\dagger}\left(\int d\boldsymbol{x}\Lambda_{j}\tilde{\Psi}_{j}^{\dagger}\right)^{r-1},
\end{equation}
This is a familiar result for functional derivative of integrals.
We note here that it is correct given our restricted functional derivative
definitions since $\tilde{\Psi}_{j}$ can be expanded in our restricted
basis set, by definition~(\ref{eqn:func-operators:restricted-psi}).

The successive application of the differential gives us
\begin{equation}
\left(\frac{\delta}{\delta\Lambda_{j}^{\prime}}\right)^{r}\left(\int d\boldsymbol{x}\Lambda_{j}\tilde{\Psi}_{j}^{\dagger}\right)^{r}=r!(\tilde{\Psi}_{j}^{\prime\dagger})^{r}.
\end{equation}
Similarly for the other differential:
\begin{equation}
\left(-\frac{\delta}{\delta\Lambda_{j}^{\prime*}}\right)^{s}\left(-\int d\boldsymbol{x}\Lambda_{j}\tilde{\Psi}_{j}^{\dagger}\right)^{s}=s!(\tilde{\Psi}_{j}^{\prime\dagger})^{s}.
\end{equation}

Thus, the differentiation will eliminate all lower order terms in
the expansion, and all higher order terms will be eliminated by setting
$\Lambda_{j}\equiv0$ for every $j$, leaving only one operator product
with the required order. \end{proof}
\begin{thm}
\label{thm:func-wigner:moments} For any non-negative integers $r_{j}$,
$s_{j}$
\begin{equation}
\langle\left\{ \prod_{j=1}^{C}\tilde{\Psi}_{j}^{r_{j}}(\tilde{\Psi}_{j}^{\dagger})^{s_{j}}\right\} _{\mathrm{sym}}\rangle=\int\delta^{2}\boldsymbol{\Psi}\left(\prod_{j=1}^{C}\Psi_{j}^{r_{j}}(\Psi_{j}^{*})^{s_{j}}\right)W[\boldsymbol{\Psi},\boldsymbol{\Psi}^{*}],
\end{equation}

where we have used the functional integration $\int\delta^{2}\boldsymbol{\Psi}$
from the Definition~\ref{def:func-calculus:func-integration}.\end{thm}
\begin{proof}
By definition of the Wigner functional, the right hand side in the
above equation can be written:
\begin{eqnarray}
I & = & \int\delta^{2}\boldsymbol{\Psi}\left(\prod_{j=1}^{C}\Psi_{j}^{r_{j}}(\Psi_{j}^{*})^{s_{j}}\right)W[\boldsymbol{\Psi},\boldsymbol{\Psi}^{*}]\nonumber \\
 & = & \frac{1}{\pi^{2\sum|\mathbb{M}_{j}|}}\mathrm{Tr}\left\{ \hat{\rho}\prod_{j=1}^{C}\int\delta^{2}\Lambda_{j}\left(\int\delta^{2}\Psi_{j}\,\Psi_{j}^{r_{j}}(\Psi_{j}^{*})^{s_{j}}D[\Lambda_{j},\Lambda_{j}^{*},\Psi_{j},\Psi_{j}^{*}]\right)\hat{D}_{j}[\Lambda_{j},\Lambda_{j}^{*}]\right\} .
\end{eqnarray}
Evaluating the integral over $\Psi_{j}$ using Lemma~\ref{lmm:func-calculus:fourier-of-moments}:
\begin{equation}
I=\mathrm{Tr}\left\{ \hat{\rho}\prod_{j=1}^{C}\int\delta^{2}\Lambda_{j}\left(\left(-\frac{\delta}{\delta\Lambda_{j}^{*}}\right)^{r_{j}}\left(\frac{\delta}{\delta\Lambda_{j}}\right)^{s_{j}}\Delta_{\mathbb{M}_{j}}[\Lambda_{j}]\right)\hat{D}_{j}[\Lambda_{j},\Lambda_{j}^{*}]\right\} ,
\end{equation}
where $\Delta_{\mathbb{M}_{j}}$ is the delta functional from the
Definition~\ref{def:func-calculus:delta-functional}. Integrating
by parts for each component in turn and eliminating terms which fit
Lemma~\ref{lmm:func-calculus:zero-delta-integrals}:
\begin{eqnarray}
I & = & \mathrm{Tr}\left\{ \hat{\rho}\prod_{j=1}^{C}\int\delta^{2}\Lambda_{j}\Delta_{\mathbb{M}_{j}}[\Lambda_{j}]\left(-\frac{\delta}{\delta\Lambda_{j}^{*}}\right)^{r_{j}}\left(\frac{\delta}{\delta\Lambda_{j}}\right)^{s_{j}}\hat{D}_{j}[\Lambda_{j},\Lambda_{j}^{*}]\right\} \nonumber \\
 & = & \left.\left(\prod_{j=1}^{C}\left(\frac{\delta}{\delta\Lambda_{j}}\right)^{s_{j}}\left(-\frac{\delta}{\delta\Lambda_{j}^{*}}\right)^{r_{j}}\right)\chi_{W}[\boldsymbol{\Lambda},\boldsymbol{\Lambda}^{*}]\right|_{\boldsymbol{\Lambda}\equiv0},
\end{eqnarray}
where $\Delta_{\mathbb{M}_{j}}$ is a delta functional from the Definition~\ref{def:func-calculus:delta-functional}.
Now, \foreignlanguage{american}{recognizing} the final expression
as a part of the previous result above in Lemma~\ref{lmm:func-wigner:moments-from-chi},
we immediately get the statement of the theorem.
\end{proof}

\section{Specific cases of transformations}

In order to Wigner transform the master equation~(\ref{eqn:master-eqn:master-eqn}),
we will need several theorems about transformations of specific operator
products. These theorems employ the expressions for high-order commutators
of restricted field operators, which look somewhat similar to those
for single-mode bosonic operators, or standard field operators from~\cite{Louisell1990}.
\begin{lem}
\label{lmm:func-operators:commutators}Commutators for restricted
field operators:
\end{lem}
\begin{equation}
\left[\tilde{\Psi},(\tilde{\Psi}^{\prime\dagger})^{l}\right]=l\delta_{\mathbb{M}}(\boldsymbol{x}^{\prime},\boldsymbol{x})(\tilde{\Psi}^{\prime\dagger})^{l-1},\qquad\left[\tilde{\Psi}^{\dagger},(\tilde{\Psi}^{\prime})^{l}\right]=-l\delta_{\mathbb{M}}^{*}(\boldsymbol{x}^{\prime},\boldsymbol{x})(\tilde{\Psi}^{\prime})^{l-1}.
\end{equation}

\begin{proof}
Proved by induction.
\end{proof}
A further generalization of these relations is
\begin{lem}
\label{lmm:func-operators:functional-commutators}
\begin{equation}
\left[\tilde{\Psi},f(\tilde{\Psi}^{\prime},\tilde{\Psi}^{\prime\dagger})\right]=\delta_{\mathbb{M}}(\boldsymbol{x}^{\prime},\boldsymbol{x})\frac{\partial f}{\partial\tilde{\Psi}^{\prime\dagger}},\qquad\left[\tilde{\Psi}^{\dagger},f(\tilde{\Psi}^{\prime},\tilde{\Psi}^{\prime\dagger})\right]=-\delta_{\mathbb{M}}^{*}(\boldsymbol{x}^{\prime},\boldsymbol{x})\frac{\partial f}{\partial\tilde{\Psi}^{\prime}},
\end{equation}
where $f(z,z^{*})$ is a function that can be expanded into the power
series of $z$ and $z^{*}$.\end{lem}
\begin{proof}
Let us prove the first relation; the procedure for the second one
is the same. Without loss of generality, we assume that $f(\tilde{\Psi}^{\prime},\tilde{\Psi}^{\prime\dagger})$
can be expanded in power series of normally ordered operators. Using
Lemma~\ref{lmm:func-operators:commutators}:

\begin{eqnarray*}
\left[\tilde{\Psi},f(\tilde{\Psi}^{\prime},\tilde{\Psi}^{\prime\dagger})\right] & = & \sum_{r,s}f_{rs}\left[\tilde{\Psi},(\tilde{\Psi}^{\prime\dagger})^{r}(\tilde{\Psi}^{\prime})^{s}\right]\\
 & = & \sum_{r,s}f_{rs}\left[\tilde{\Psi},(\tilde{\Psi}^{\prime\dagger})^{r}\right](\tilde{\Psi}^{\prime})^{s}\\
 & = & \sum_{r,s}f_{rs}r\delta_{\mathbb{M}}(\boldsymbol{x}^{\prime},\boldsymbol{x})(\tilde{\Psi}^{\prime\dagger})^{r-1}(\tilde{\Psi}^{\prime})^{s}\\
 & = & \delta_{\mathbb{M}}(\boldsymbol{x}^{\prime},\boldsymbol{x})\frac{\partial f}{\partial\tilde{\Psi}^{\prime\dagger}}.
\end{eqnarray*}

\end{proof}
The simplest case is the transformation of the linear part of the
Hamiltonian~(\ref{eqn:master-eqn:hamiltonian}).
\begin{thm}
\label{thm:transformations:w-commutator1}
\begin{equation}
\mathcal{W}\left[[\int d\boldsymbol{x}\tilde{\Psi}_{j}^{\dagger}\tilde{\Psi}_{k},\hat{A}]\right]=\int d\boldsymbol{x}\left(-\frac{\delta}{\delta\Psi_{j}}\Psi_{k}+\frac{\delta}{\delta\Psi_{k}^{*}}\Psi_{j}^{*}\right)\mathcal{W}[\hat{A}].
\end{equation}
\end{thm}
\begin{proof}
Proved straightforwardly using Theorem~\ref{thm:func-wigner:correspondences}
and the relation
\begin{equation}
\Psi_{k}\frac{\delta}{\delta\Psi_{j}}\mathcal{F}=\left(\frac{\delta}{\delta\Psi_{j}}\Psi_{k}-\delta_{jk}\delta_{\mathbb{M}_{j}}(\boldsymbol{x},\boldsymbol{x})\right)\mathcal{F}.
\end{equation}

\end{proof}
The expression $\delta_{\mathbb{M}_{j}}(\boldsymbol{x},\boldsymbol{x})$
will appear in many expressions later in the paper, so we will denote
$\tilde{\delta}_{j}\equiv\delta_{\mathbb{M}_{j}}(\boldsymbol{x},\boldsymbol{x})$
for brevity.

Commutators with a Laplacian inside require a somewhat special treatment,
because in general the Laplacian acts on basis functions. For our
purposes we only need one specific case, and, fortunately, in this
case the Laplacian behaves like a constant.
\begin{thm}
\label{thm:transformations:w-laplacian-commutator1}
\begin{equation}
\mathcal{W}\left[\int d\boldsymbol{x}[\tilde{\Psi}^{\dagger}\nabla^{2}\tilde{\Psi},\hat{A}]\right]=\int d\boldsymbol{x}\left(-\frac{\delta}{\delta\Psi}\nabla^{2}\Psi+\frac{\delta}{\delta\Psi^{*}}\nabla^{2}\Psi^{*}\right)\mathcal{W}[\hat{A}].
\end{equation}
\end{thm}
\begin{proof}
Proved using Theorem~\ref{thm:func-wigner:correspondences} and Lemma~\ref{lmm:func-calculus:move-laplacian}.
\end{proof}
The next theorem describes the transformation of the non-linear part
of the Hamiltonian~(\ref{eqn:master-eqn:effective-H}).
\begin{thm}
\label{thm:transformations:w-commutator2}
\begin{eqnarray}
 &  & \mathcal{W}\left[[\int d\boldsymbol{x}\tilde{\Psi}_{j}^{\dagger}\tilde{\Psi}_{k}^{\dagger}\tilde{\Psi}_{j}\tilde{\Psi}_{k},\hat{A}]\right]\nonumber \\
 &  & =\int d\boldsymbol{x}\left(\frac{\delta}{\delta\Psi_{j}}\left(-\Psi_{j}\Psi_{k}\Psi_{k}^{*}+\frac{\tilde{\delta}_{k}}{2}\left(\delta_{jk}\Psi_{k}+\Psi_{j}\right)\right)\right.\nonumber \\
 &  & \left.+\frac{\delta}{\delta\Psi_{j}^{*}}\left(\Psi_{j}^{*}\Psi_{k}\Psi_{k}^{*}-\frac{\tilde{\delta}_{k}}{2}\left(\delta_{jk}\Psi_{k}^{*}+\Psi_{j}^{*}\right)\right)\right.\nonumber \\
 &  & \left.+\frac{\delta}{\delta\Psi_{k}}\left(-\Psi_{j}\Psi_{j}^{*}\Psi_{k}+\frac{\tilde{\delta}_{j}}{2}\left(\delta_{jk}\Psi_{j}+\Psi_{k}\right)\right)\right.\nonumber \\
 &  & \left.+\frac{\delta}{\delta\Psi_{k}^{*}}\left(\Psi_{j}\Psi_{j}^{*}\Psi_{k}^{*}-\frac{\tilde{\delta}_{j}}{2}\left(\delta_{jk}\Psi_{j}^{*}+\Psi_{k}^{*}\right)\right)\right.\nonumber \\
 &  & \left.+\frac{\delta}{\delta\Psi_{j}}\frac{\delta}{\delta\Psi_{j}^{*}}\frac{\delta}{\delta\Psi_{k}}\frac{1}{4}\Psi_{k}-\frac{\delta}{\delta\Psi_{j}}\frac{\delta}{\delta\Psi_{j}^{*}}\frac{\delta}{\delta\Psi_{k}^{*}}\frac{1}{4}\Psi_{k}^{*}\right.\nonumber \\
 &  & \left.+\frac{\delta}{\delta\Psi_{k}}\frac{\delta}{\delta\Psi_{k}^{*}}\frac{\delta}{\delta\Psi_{j}}\frac{1}{4}\Psi_{j}-\frac{\delta}{\delta\Psi_{k}}\frac{\delta}{\delta\Psi_{k}^{*}}\frac{\delta}{\delta\Psi_{j}^{*}}\frac{1}{4}\Psi_{j}^{*}\right)\mathcal{W}[\hat{A}].
\end{eqnarray}
 \end{thm}
\begin{proof}
The proof is the same as in the case of Theorem~\ref{thm:transformations:w-commutator1}.
\end{proof}
Finally, the transformation of loss terms~(\ref{eqn:master-eqn:loss-term})
requires some treatment. The proof makes use of two auxiliary lemmas.
The first one will help us move functional differentials to their
intended places (namely, to the left).
\begin{lem}
\label{lmm:transformations:swap-differential} For $\mathcal{F}\in\mathbb{F}_{\mathbb{M}}\rightarrow\mathbb{F}$
and any non-negative integer $a$, $b$:
\begin{eqnarray}
 &  & \Psi^{a}\left(\frac{\delta}{\delta\Psi}\right)^{b}\mathcal{F}[\Psi,\Psi^{*}]\nonumber \\
 &  & =\sum_{j=0}^{\min(a,b)}\binom{b}{j}\frac{(-1)^{j}a!}{(a-j)!}\tilde{\delta}^{j}\left(\frac{\delta}{\delta\Psi}\right)^{b-j}\Psi^{a-j}\mathcal{F}[\Psi,\Psi^{*}]
\end{eqnarray}
\end{lem}
\begin{proof}
Proved straightforwardly by induction.
\end{proof}
The second lemma gives a way to simplify sums obtained from the application
of the previous lemma.
\begin{lem}[Sum rearrangement]
\label{lmm:transformations:sum-rearrangement}For any non-negative
integer $l$, $u$:
\begin{equation}
\sum_{j=0}^{l}\sum_{k=0}^{\min(l-u,j)}x^{j-k}Q(j,k)=\sum_{v=0}^{l}x^{v}\sum_{k=0}^{l-\max(u,v)}Q(v+k,k).
\end{equation}
\end{lem}
\begin{proof}
Obviously, the order $v=j-k$ of factor $f$ can vary from $0$ (say,
when $j=0$ and $k=0$) to $l$ (when $j=l$ and $k=0$). Therefore:

\[
\sum_{j=0}^{l}\sum_{k=0}^{\min(l-u,j)}f^{j-k}g(j,k)=\sum_{v=0}^{l}f^{v}\sum_{k\in K(l,u,v)}g(v+k,k),
\]
where the set $K$ is defined as
\begin{eqnarray*}
K(l,u,v) & = & \{k|0\le j\le l\wedge0\le k\le\min(l-u,j)\wedge j-k=v\}\\
 & = & \{k|k\le l-v\wedge0\le k\le\min(l-u,v+k)\}.
\end{eqnarray*}

It is convenient to consider two cases separately $v\le u$ and $v>u$.
For the former case
\[
K_{v\le u}=\{k|k\le l-v\wedge0\le k\le\min(l-u,k+v)\wedge v\le u\}.
\]
When $v\le u$, $k\le l-v\le l-u\le\min(l-u,k+v)$ is always true,
and the first inequation is redundant:

\[
K_{v\le u}=\{k|0\le k\le\min(l-u,v+k)\wedge v\le u\}.
\]
Splitting into two sets to get rid of the minimum function:

\begin{eqnarray*}
K_{v\le u} & = & \{k|v\le u\wedge k\ge0\wedge((k\le l-u\wedge l-u<v+k)\vee(k\le v+k\wedge l-u\ge v+k))\}\\
 & = & \{k|v\le u\wedge0\le k\le l-u\}.
\end{eqnarray*}
For the latter case:

\begin{eqnarray*}
K_{v>u} & = & \{k|k\le l-v\wedge0\le k\le\min(l-u,k+v)\wedge v>u\}\\
 & = & \{k|v>u\wedge k\ge0\wedge((k\le l-v\wedge k\le l-u\wedge l-u\le k+v)\\
 &  & \vee(k\le l-v\wedge k\le k+v\wedge l-u>k+v))\}\\
 & = & \{k|v>u\wedge0\le k\le l-v\}.
\end{eqnarray*}

Thus

\begin{eqnarray*}
K & = & K_{v\le u}\cup K_{v>u}\\
 & = & \{k|v\le u\wedge0\le k\le l-u\}\cup\{k|v>u\wedge0\le k\le l-v\}\\
 & = & \{k|0\le k\le l-\max(u,v)\},
\end{eqnarray*}
which gives us the statement of the lemma.
\end{proof}
Finally, the loss transformation theorem can be proved.
\begin{thm}
\label{thm:transformations:w-losses} The Wigner transformation of
loss term~(\ref{eqn:master-eqn:loss-term}) is
\begin{equation}
\mathcal{W}\left[\int d\boldsymbol{x}\mathcal{L}_{\boldsymbol{l}}[\hat{A}]\right]=\int d\boldsymbol{x}\sum_{j_{1}=0}^{l_{1}}\sum_{k_{1}=0}^{l_{1}}\ldots\sum_{j_{C}=0}^{l_{C}}\sum_{k_{C}=0}^{l_{C}}\left(\prod_{c=1}^{C}\left(\frac{\delta}{\delta\Psi_{c}^{*}}\right)^{j_{c}}\left(\frac{\delta}{\delta\Psi_{c}}\right)^{k_{c}}\right)L_{\boldsymbol{l},\boldsymbol{j},\boldsymbol{k}}\mathcal{W}[\hat{A}],
\end{equation}
where the nonlinear loss coefficient $L$ is
\begin{eqnarray}
L_{\boldsymbol{l},\boldsymbol{j},\boldsymbol{k}}= &  & \left(2-(-1)^{\sum_{c}j_{c}}-(-1)^{\sum_{c}k_{c}}\right)\times\nonumber \\
 &  & \times\prod_{c=1}^{C}\left(\sum_{m_{c}=0}^{l_{c}-\max(j_{c},k_{c})}Q(l_{c},j_{c},k_{c},m_{c})\delta_{\mathbb{M}_{c}}^{m_{c}}(\boldsymbol{x},\boldsymbol{x})\Psi_{c}^{l_{c}-j_{c}-m_{c}}(\Psi_{c}^{*})^{l_{c}-k_{c}-m_{c}}\right),
\end{eqnarray}
and we introduce a numerical factor $Q$, where
\begin{equation}
Q(l,j,k,m)=\frac{(-1)^{m}}{2^{j+k+m}}\frac{(l!)^{2}}{m!j!k!(l-k-m)!(l-j-m)!}.
\end{equation}
\end{thm}
\begin{proof}
Proved by applying Theorem~\ref{thm:func-wigner:correspondences},
expanding products using binomial theorem, using Lemma~\ref{lmm:transformations:swap-differential}
to move differentials to front, and applying Lemma~\ref{lmm:transformations:sum-rearrangement}
to transform the resulting summations.
\end{proof}

\section{Wigner truncation and Fokker-Planck equation\label{sec:Wigner-truncation}}

Now we have all necessary tools to transform the master equation~(\ref{eqn:master-eqn:master-eqn})
with the Wigner transformation from Definition~\ref{def:func-wigner:w-transformation}
to the form of a partial differential equation.

The single-particle term~(\ref{eqn:master-eqn:single-particle})
is transformed using Theorem~\ref{thm:transformations:w-commutator1}
and Theorem~\ref{thm:transformations:w-laplacian-commutator1} (since
$K_{j}$ is basically a sum of Laplacian operator and functions of
$\boldsymbol{x}$):
\begin{equation}
\mathcal{W}\left[[\int d\boldsymbol{x}\tilde{\Psi}_{j}^{\dagger}K_{jk}\tilde{\Psi}_{k},\hat{\rho}]\right]=\int d\boldsymbol{x}\left(-\frac{\delta}{\delta\Psi_{j}}K_{jk}\Psi_{k}+\frac{\delta}{\delta\Psi_{k}^{*}}K_{jk}\Psi_{j}^{*}\right)W,
\end{equation}
where the Wigner function $W\equiv\mathcal{W}[\hat{\rho}]$. The nonlinear
term is transformed with Theorem~\ref{thm:transformations:w-commutator2}
(assuming $U_{kj}=U_{jk}$):
\begin{eqnarray}
\mathcal{W}\left[[\int d\boldsymbol{x}\frac{U_{jk}}{2}\tilde{\Psi}_{j}^{\dagger}\tilde{\Psi}_{k}^{\dagger}\tilde{\Psi}_{j}\tilde{\Psi}_{k},\hat{\rho}]\right] & = & \int d\boldsymbol{x}U_{jk}\left(\frac{\delta}{\delta\Psi_{j}}\left(-\Psi_{j}\Psi_{k}\Psi_{k}^{*}+\frac{\tilde{\delta}_{k}}{2}(\delta_{jk}\Psi_{k}+\Psi_{j})\right)\right.\nonumber \\
 &  & \left.+\frac{\delta}{\delta\Psi_{j}^{*}}\left(\Psi_{j}^{*}\Psi_{k}\Psi_{k}^{*}-\frac{\tilde{\delta}_{k}}{2}(\delta_{jk}\Psi_{k}^{*}+\Psi_{j}^{*})\right)\right.\nonumber \\
 &  & \left.+\frac{\delta}{\delta\Psi_{j}}\frac{\delta}{\delta\Psi_{j}^{*}}\frac{\delta}{\delta\Psi_{k}}\frac{1}{4}\Psi_{k}-\frac{\delta}{\delta\Psi_{j}}\frac{\delta}{\delta\Psi_{j}^{*}}\frac{\delta}{\delta\Psi_{k}^{*}}\frac{1}{4}\Psi_{k}^{*}\right)W.\label{eqn:truncation:nonlinear-term}
\end{eqnarray}
Loss terms~(\ref{eqn:master-eqn:loss-term}) are transformed with
Theorem~\ref{thm:transformations:w-losses} and result in a similar
equation, with a finite number of differential terms up to order $2n$
for $n-$body collisional losses. It should be recalled that in the
above equation, the notation $\tilde{\delta}_{j}\equiv\delta_{\mathbb{M}_{j}}(\boldsymbol{x},\boldsymbol{x})$
was introduced previously for brevity; in general this is a cut-off
dependent constant with units of density.

Assuming that $K_{jk}$, $U_{jk}$, and $\kappa_{\boldsymbol{l}}$
are real-valued, all the transformations described above result in
a partial differential equation for $W$ of the form
\begin{equation}
\frac{\partial W}{\partial t}=\int d\boldsymbol{x}\left\{ -\sum_{j=1}^{C}\frac{\delta}{\delta\Psi_{j}}\mathcal{A}_{j}-\sum_{j=1}^{C}\frac{\delta}{\delta\Psi_{j}^{*}}\mathcal{A}_{j}^{*}+\sum_{j=1}^{C}\sum_{k=1}^{C}\frac{\delta^{2}}{\delta\Psi_{j}^{*}\delta\Psi_{k}}\mathcal{D}_{jk}+\mbox{O}\left[\frac{\delta^{3}}{\delta\Psi_{j}^{3}}\right]\right\} W.
\end{equation}
Terms of order higher than 2 are produced both by the nonlinear term
in the Hamiltonian and loss terms. Such an equation could be solved
perturbatively if there were only orders up to 3 (which means an absence
of nonlinear losses)~\cite{Polkovnikov2003}, but in most cases all
terms except for first- and second-order ones are truncated. In order
to justify this truncation in a consistent way, we develop an order-by-order
expansion in $1/N_{c}$, where $N_{c}$ is a characteristic particle
number in a physical interaction volume, and truncate terms of formal
order $1/N_{c}^{2}$. This is achieved~\cite{Drummond1993} by use
of the formal definition of a scaled Wigner function $W^{\psi}$,
satisfying a scaled equation in terms of dimensionless scaled fields
$\psi$, with:
\begin{eqnarray}
\psi_{j} & = & \Psi_{j}\sqrt{\ell_{c}^{D}/N_{c}}\nonumber \\
\mathcal{A}_{j}^{\psi} & = & t_{c}\sqrt{\ell_{c}^{D}/N_{c}}\mathcal{A}_{j}+\mbox{O}\left(1/N_{c}^{2}\right)\nonumber \\
\mathcal{D}_{jk}^{\psi} & = & t_{c}\left(\ell_{c}^{D}/N_{c}\right)\mathcal{D}_{jk}+\mbox{O}\left(1/N_{c}^{2}\right).
\end{eqnarray}
Here $t_{c}$ is a characteristic interaction time and $\ell_{c}$
is a characteristic interaction length. These would normally be chosen
as the healing time and healing length respectively in a BEC calculation.
Typically the cell size is chosen as proportional to the healing length,
for optimum accuracy in resolving spatial detail. Using this expansion,
a consistent order-by-order expansion in $(1/N_{c})$ can be obtained,
of form:
\begin{equation}
\frac{\partial W^{\psi}}{\partial\tau}=\int d\boldsymbol{x}\left\{ -\sum_{j=1}^{C}\frac{\delta}{\delta\psi_{j}}\mathcal{A}_{j}^{\psi}-\sum_{j=1}^{C}\frac{\delta}{\delta\psi_{j}^{*}}\mathcal{A}_{j}^{\psi*}+\sum_{j=1}^{C}\sum_{k=1}^{C}\frac{\delta^{2}}{\delta\psi_{j}^{*}\delta\psi_{k}}\mathcal{D}_{jk}^{\psi}+\mbox{O}\left[\frac{1}{N_{c}^{2}}\right]\right\} W^{\psi}.
\end{equation}

With the assumption of the state being coherent, the simple condition
for truncation \textemdash{} i.e., omitting terms of $O(1/N_{c}^{2})$
\textemdash{} can be shown to be~\cite{Sinatra2002}
\begin{equation}
N_{j}\gg|\mathbb{M}_{j}|,
\end{equation}
where $N_{j}$ is the total number of atoms of the component $j$.
The inclusion of the mode factor is caused by the fact that the number
of additional terms increases as the number of modes increases, which
may be needed to treat convergence of the method for large momentum
cutoff. We see immediately that there are subtleties involved if one
wishes to include larger numbers of high-momentum modes, since this
increases the mode number while leaving the numbers unchanged. In
other words, the truncation technique is inherently restricted in
its ability to resolve fine spatial details in the high-momentum cutoff
limit.

The $1/N_{c}$ is equivalent to an expansion in the inverse density,
which requires the inequality~\cite{Norrie2006}
\begin{equation}
\tilde{\delta}_{j}=\delta_{\mathbb{M}_{j}}(\boldsymbol{x},\boldsymbol{x})\ll|\Psi_{j}|^{2}.\label{eqn:truncation:delta-condition}
\end{equation}
The coherency assumption does not, of course, encompass all possible
states that can be produced during evolution, which means that the
condition above is more of a guide than a restriction. For certain
systems the truncation was shown to work even when~\eqref{eqn:truncation:delta-condition}
is violated~\cite{Ruostekoski2005}. The validity may also depend
on the simulation time~\cite{Javanainen2013}, and other physically
relevant factors.

A common example of such relevant factors is that there can be a large
difference in the size of the original parameters. To illustrate this
issue, one may have a situation where $\kappa_{1}\approx\kappa_{2}N_{c}$
even though $N_{c}\gg1$ . Under these conditions, it is essential
to include a scaling of the parameters in calculating the formal order,
so that the scaled parameters have comparable sizes. This allows one
to correctly identify which terms are negligible in a given physical
problem, and which terms must be included.

In general, one can estimate the validity of truncation for the particular
problem and the particular observable by calculating the quantum correction~\cite{Polkovnikov2010}.
Other techniques for estimating validity include comparison with the
exact positive-P simulation method~\cite{Drummond1993}, and examining
results for unphysical behaviour such as negative occupation numbers~\cite{Deuar2007}.
It is generally the case for unitary evolution that errors caused
by truncation grow in time, leading to a finite time horizon for applicability,
as explained in the introduction.

The use of this Wigner truncation allows us to simplify the results
of Theorem~\ref{thm:transformations:w-commutator2} and Theorem~\ref{thm:transformations:w-losses}.
Wigner truncation is an expansion up to the order $1/N_{c}$, so during
the simplification, along with the higher order derivatives, we drop
all components with $\tilde{\delta}_{j}$ of order higher than 1 in
the drift terms, and of order higher than 0 in the diffusion terms.
\begin{lem}
Assuming the conditions for Wigner truncation are satisfied, the result
of Wigner transformation of the nonlinear term can be written as

\begin{eqnarray*}
\mathcal{W}\left[[\frac{U_{jk}}{2}\tilde{\Psi}_{j}^{\dagger}\tilde{\Psi}_{k}^{\dagger}\tilde{\Psi}_{j}\tilde{\Psi}_{k},\hat{\rho}]\right] & \approx & U_{jk}\left(\frac{\delta}{\delta\Psi_{j}}\left(-\Psi_{j}\Psi_{k}\Psi_{k}^{*}+\frac{\tilde{\delta}_{k}}{2}(\delta_{jk}\Psi_{k}+\Psi_{j})\right)\right.\\
 &  & \left.+\frac{\delta}{\delta\Psi_{j}^{*}}\left(\Psi_{j}^{*}\Psi_{k}\Psi_{k}^{*}-\frac{\tilde{\delta}_{k}}{2}(\delta_{jk}\Psi_{k}^{*}+\Psi_{j}^{*})\right)\right)W
\end{eqnarray*}
\end{lem}
\begin{proof}
Proved by simplifying equation (\ref{eqn:truncation:nonlinear-term})
under the Wigner truncation condition (essentially by dropping terms
with third order derivatives).\end{proof}
\begin{lem}
Assuming the conditions for Wigner truncation are satisfied, the result
of Wigner transformation of the loss term can be written as
\begin{eqnarray}
\mathcal{W}[\mathcal{L}_{\boldsymbol{l}}[\hat{\rho}]] & \approx & \left(\sum_{j=1}^{C}\frac{\delta}{\delta\Psi_{j}^{*}}\left(\frac{\partial O_{\boldsymbol{l}}}{\partial\Psi_{j}}O_{\boldsymbol{l}}^{*}-\frac{1}{2}\sum_{k=1}^{C}\tilde{\delta}_{k}\frac{\partial^{2}O_{\boldsymbol{l}}}{\partial\Psi_{j}\partial\Psi_{k}}\frac{\partial O_{\boldsymbol{l}}^{*}}{\partial\Psi_{k}^{*}}\right)\right.\nonumber \\
 &  & +\sum_{j=1}^{C}\frac{\delta}{\delta\Psi_{j}}\left(\frac{\partial O_{\boldsymbol{l}}^{*}}{\partial\Psi_{j}^{*}}O_{\boldsymbol{l}}-\frac{1}{2}\sum_{k=1}^{C}\tilde{\delta}_{k}\frac{\partial^{2}O_{\boldsymbol{l}}^{*}}{\partial\Psi_{j}^{*}\partial\Psi_{k}^{*}}\frac{\partial O_{\boldsymbol{l}}}{\partial\Psi_{k}}\right)\nonumber \\
 &  & \left.+\sum_{j=1}^{C}\sum_{k=1}^{C}\frac{\delta^{2}}{\delta\Psi_{j}^{*}\delta\Psi_{k}}\frac{\partial O_{\boldsymbol{l}}}{\partial\Psi_{j}}\frac{\partial O_{\boldsymbol{l}}^{*}}{\partial\Psi_{k}^{*}}\right)W
\end{eqnarray}
where $O_{\boldsymbol{l}}\equiv O_{\boldsymbol{l}}[\boldsymbol{\Psi}]=\prod_{j=1}^{C}\Psi_{j}^{l_{j}}$. \end{lem}
\begin{proof}
The proof is basically a simplification of the result of Theorem~\ref{thm:transformations:w-losses}
under two conditions. First, we neglect all terms with order lower
than $1/N$. This means that we are only considering terms with $\sum m_{c}\le1$
in the drift part, and $\sum m_{c}=0$ in the diffusion part. Second,
we are dropping all terms with high order differentials, which can
be expressed as limiting $\sum j_{c}+\sum k_{c}\le2$. The only combinations
of $j_{c}$ and $k_{c}$ for which $Z(\boldsymbol{j},\boldsymbol{k})$
is not zero are thus $\{j_{c}=\delta_{cn},k_{c}=0,n\in[1,C]\}$, $\{j_{c}=0,k_{c}=\delta_{cn},n\in[1,C]\}$
and $\{j_{c}=\delta_{cn},k_{c}=\delta_{cp},n\in[1,C],p\in[1,C]\}$.
These combinations produce terms with $\delta/\delta\Psi_{n}^{*}$,
$\delta/\delta\Psi_{n}$ (drift) and $\delta^{2}/\delta\Psi_{p}\delta\Psi_{n}^{*}$
(diffusion) respectively. Applying these conditions one can get the
statement of the theorem.
\end{proof}
Thus the truncated Fokker-Planck equation (FPE) is:
\begin{equation}
\frac{dW}{dt}=\int d\boldsymbol{x}\left(-\sum_{j=1}^{C}\frac{\delta}{\delta\Psi_{j}}\mathcal{A}_{j}-\sum_{j=1}^{C}\frac{\delta}{\delta\Psi_{j}^{*}}\mathcal{A}_{j}^{*}+\sum_{j=1}^{C}\sum_{k=1}^{C}\frac{\delta^{2}}{\delta\Psi_{j}^{*}\delta\Psi_{k}}\mathcal{D}_{jk}\right)W,\label{eq:truncation:FPE}
\end{equation}
or, in matrix form:

\[
\frac{dW}{dt}=\int d\boldsymbol{x}\left(-2\mathrm{Re}\left(\boldsymbol{\delta}_{\boldsymbol{\Psi}}\cdot\boldsymbol{\mathcal{A}}\right)+\mathrm{Tr}\left\{ \boldsymbol{\delta}_{\boldsymbol{\Psi}^{*}}\boldsymbol{\delta}_{\boldsymbol{\Psi}}^{T}\mathcal{D}\right\} \right)W,
\]
where we define the relevant coefficients in the FPE as:
\begin{eqnarray}
\mathcal{A}_{j} & = & -\frac{i}{\hbar}\left(\sum_{k=1}^{C}K_{jk}\Psi_{k}+\Psi_{j}\sum_{k=1}^{C}U_{jk}\left(\vert\Psi_{k}\vert^{2}-\frac{\delta_{jk}+1}{2}\tilde{\delta}_{k}\right)\right)\nonumber \\
 &  & -\sum_{\boldsymbol{l}}\kappa_{\boldsymbol{l}}\left(\frac{\partial O_{\boldsymbol{l}}^{*}}{\partial\Psi_{j}^{*}}O_{\boldsymbol{l}}-\frac{1}{2}\sum_{k=1}^{C}\tilde{\delta}_{k}\frac{\partial^{2}O_{\boldsymbol{l}}^{*}}{\partial\Psi_{j}^{*}\partial\Psi_{k}^{*}}\frac{\partial O_{\boldsymbol{l}}}{\partial\Psi_{k}}\right),\label{eq:drift-term}
\end{eqnarray}
and
\begin{equation}
\mathcal{D}_{jk}=\sum_{\boldsymbol{l}}\kappa_{\boldsymbol{l}}\frac{\partial O_{\boldsymbol{l}}}{\partial\Psi_{j}}\frac{\partial O_{\boldsymbol{l}}^{*}}{\partial\Psi_{k}^{*}}.
\end{equation}

\section{Stochastic differential equations}

Direct solution of the above FPE is generally impractical, and a Monte-Carlo
or sampled calculation is called for. Since the diffusion matrix is
positive-definite, the truncated Wigner function $W$ is a probability
distribution, provided it has a positive initial distribution. Therefore
the equation can be further transformed to the equivalent set of stochastic
differential equations in Itô form.

\subsection{Stochastic Evolution}

General results on such transformations are given in Appendix~\ref{sec:app-functional-fpe},
as described by Theorem~\ref{thm:app-fpe:fpe-sde-func}. Application
of these methods to the truncated FPE~(\ref{eq:truncation:FPE})
gives immediately a system of SDEs in Itô~\cite{Gardiner1997} form:

\begin{equation}
d\boldsymbol{\Psi}=\boldsymbol{\mathcal{P}}\left[\boldsymbol{\mathcal{A}}dt+\mathcal{B}d\boldsymbol{Q}\right],\label{eqn:fpe:sdes}
\end{equation}
where the drift term $\boldsymbol{\mathcal{A}}$ is given by \eqref{eq:drift-term},
and noise term is a matrix with elements
\begin{equation}
\mathcal{B}_{j\boldsymbol{l}}=\sqrt{\kappa_{\boldsymbol{l}}}\frac{\partial O_{\boldsymbol{l}}^{*}}{\partial\Psi_{j}^{*}}.
\end{equation}
Here $Q_{\boldsymbol{l}}$ is a functional Wiener process:
\begin{equation}
Q_{\boldsymbol{l}}=\sum_{\boldsymbol{n}\in\mathbb{B}}\phi_{j}Z_{\boldsymbol{l},\boldsymbol{n}},
\end{equation}
and $Z_{\boldsymbol{l},\boldsymbol{n}}$ are, in turn, independent
complex-valued Wiener processes with $\langle Z_{\boldsymbol{l},\boldsymbol{n}}Z_{\boldsymbol{k},\boldsymbol{m}}^{*}\rangle=\delta_{\boldsymbol{l},\boldsymbol{k}}\delta_{\boldsymbol{n},\boldsymbol{m}}dt$.

Alternatively, in Stratonovich form the SDEs look like
\begin{equation}
d\boldsymbol{\Psi}=\boldsymbol{\mathcal{P}}\left[(\boldsymbol{\mathcal{A}}-\boldsymbol{\mathcal{S}})dt+\mathcal{B}d\boldsymbol{Q}\right],
\end{equation}
where the Stratonovich~\cite{Gardiner1997} term has components
\begin{equation}
\mathcal{S}_{j}=\frac{1}{2}\sum_{n=1}^{C}\sum_{\boldsymbol{l}}\kappa_{\boldsymbol{l}}\frac{\partial O_{\boldsymbol{l}}}{\partial\Psi_{n}}\left(\frac{\partial^{2}O_{\boldsymbol{l}}}{\partial\Psi_{n}\partial\Psi_{j}}\right)^{*}\delta_{\mathbb{M}_{n}}(\boldsymbol{x},\boldsymbol{x}).
\end{equation}

These equations can now be solved using conventional methods~\cite{Blakie2008},
and any required expectations of symmetrically ordered operator products
can be obtained from their solution using Theorem~\ref{thm:func-wigner:moments}:
\begin{eqnarray}
\langle\left\{ \prod_{j=1}^{C}\tilde{\Psi}_{j}^{r_{j}}(\tilde{\Psi}_{j}^{\dagger})^{s_{j}}\right\} _{\mathrm{sym}}\rangle &  & =\int\delta\boldsymbol{\Psi}\left(\prod_{j=1}^{C}\Psi_{j}^{r_{j}}(\Psi_{j}^{*})^{s_{j}}\right)W\nonumber \\
 &  & \approx\left\langle \prod_{j=1}^{C}\Psi_{j}^{r_{j}}(\Psi_{j}^{*})^{s_{j}}\right\rangle _{\mathrm{paths}},\label{eq:expectation}
\end{eqnarray}
where $r_{c}$ and $s_{c}$ is some set of non-negative integers,
and $\left\langle \right\rangle _{\mathrm{paths}}$ stands for the
average over the simulation paths.

\subsection{Single-component example}

To illustrate the application of the theorems above to some specific
problems we will first consider a simple case with a single component
BEC, with 3-body loss and no unitary evolution (the same as described
by Norrie et al.~\cite{Norrie2006a}). For this system we have $\hat{K}\equiv0$,
$U\equiv0$ and $\hat{O}=\tilde{\Psi}^{3}$ (and, consequently, $O=\Psi^{3}$),
and we also denote $\gamma=6\kappa$. The FPE for this system is therefore

\begin{eqnarray*}
\frac{dW}{dt} & = & -\frac{\delta}{\delta\Psi}\left(-\frac{\gamma}{2}\vert\Psi\vert^{4}\Psi+\frac{3\gamma}{2}\vert\Psi\vert^{2}\Psi\tilde{\delta}-\frac{3\gamma}{4}\Psi\tilde{\delta}{}^{2}\right)-\frac{\delta}{\delta\Psi^{*}}\left(-\frac{\gamma}{2}\vert\Psi\vert^{4}\Psi^{*}+\frac{3\gamma}{2}\vert\Psi\vert^{2}\Psi^{*}\tilde{\delta}-\frac{3\gamma}{4}\Psi^{*}\tilde{\delta}{}^{2}\right)\\
 &  & +\frac{\delta^{2}}{\delta\Psi^{*}\delta\Psi}\left(\frac{3\gamma}{2}\vert\Psi\vert^{4}-3\gamma|\Psi^{2}\vert\tilde{\delta}+\frac{3\gamma}{4}\tilde{\delta}{}^{2}\right)\\
 &  & +\frac{\delta^{3}}{\delta\Psi^{*}\delta\Psi^{2}}\left(\frac{3\gamma}{8}\vert\Psi\vert^{2}\Psi-\frac{3\gamma}{8}\Psi\tilde{\delta}\right)+\frac{\delta^{3}}{\delta\Psi^{*2}\delta\Psi}\left(\frac{3\gamma}{8}\vert\Psi\vert^{2}\Psi^{*}-\frac{3\gamma}{8}\Psi^{*}\tilde{\delta}\right)\\
 &  & +\frac{\delta}{\delta\Psi^{3}}\left(\frac{\gamma}{24}\Psi^{3}\right)+\frac{\delta^{3}}{\delta\Psi^{*3}}\left(\frac{\gamma}{24}\Psi^{*3}\right)+\mathrm{O}[\frac{1}{N_{c}^{4}}].
\end{eqnarray*}
After the truncation, the resulting stochastic equation describing
the system is

\begin{eqnarray*}
d\Psi & = & \mathcal{P}\left[-\frac{\gamma}{6}\left(\frac{\partial O^{*}}{\partial\Psi^{*}}O-\frac{1}{2}\tilde{\delta}\frac{\partial^{2}O^{*}}{\partial(\Psi^{*})^{2}}\frac{\partial O}{\partial\Psi}\right)dt+\sqrt{\frac{\gamma}{6}}\frac{\partial O^{*}}{\partial\Psi^{*}}dQ(\boldsymbol{x},t)\right]\\
 & = & \mathcal{P}\left[-\left(\frac{\gamma}{2}\vert\Psi\vert^{4}\Psi-\frac{3\gamma}{2}\tilde{\delta}\vert\Psi\vert^{2}\Psi\right)dt+\sqrt{\frac{3\gamma}{2}}(\Psi^{*})^{2}dQ(\boldsymbol{x},t)\right].
\end{eqnarray*}
The equation coincides with the one given by Norrie et al., except
for the additional correction to the drift term, which is of order
$1/N$ and therefore cannot be omitted.

If we calculate the rate population change over time using Itô formula
(either by expanding $\Psi$ in mode form, or using the functional
equivalent of Itô formula), we obtain

\[
\frac{dN}{dt}=\frac{d\langle\tilde{\Psi}^{\dagger}\tilde{\Psi}\rangle}{dt}=\frac{d\langle\Psi^{*}\Psi\rangle_{\mathrm{paths}}}{dt}=-\gamma\int d\boldsymbol{x}\left(\langle\vert\Psi\vert^{6}\rangle_{\mathrm{paths}}-\frac{9}{2}\tilde{\delta}\langle\vert\Psi\vert^{4}\rangle_{\mathrm{paths}}\right).
\]
This can be transformed further to more conventional form. Using the
equivalence~\eqref{eq:expectation} and the analogue of the ordering
transformation formula~\cite{Cahill1969} for field operators

\[
\left\{ \left(\tilde{\Psi}^{\dagger}\right)^{r}\tilde{\Psi}^{s}\right\} _{\mathrm{sym}}=\sum_{k=0}^{\min(r,s)}\frac{k!}{2^{k}}\begin{pmatrix}r\\
k
\end{pmatrix}\begin{pmatrix}s\\
k
\end{pmatrix}\left(\tilde{\Psi}^{\dagger}\right)^{r-k}\tilde{\Psi}^{s-k}\tilde{\delta}^{k},
\]
we get

\[
\langle\vert\Psi\vert^{4}\rangle_{\mathrm{paths}}=g^{(2)}n^{2}+2\tilde{\delta}n+\frac{1}{2}\tilde{\delta}^{2},
\]

\[
\langle\vert\Psi\vert^{6}\rangle_{W}=g^{(3)}n^{3}+\frac{9}{2}\tilde{\delta}g^{(2)}n^{2}+\frac{9}{2}\tilde{\delta}^{2}n+\frac{3}{4}\tilde{\delta}^{3}.
\]
Here $n=\langle\tilde{\Psi}^{\dagger}\tilde{\Psi}\rangle$ is the
particle density, and $g^{(k)}=\langle\left(\tilde{\Psi}^{\dagger}\right)^{k}\tilde{\Psi}^{k}\rangle/\langle\tilde{\Psi}^{\dagger}\tilde{\Psi}\rangle$
are correlation factors. Substituting above expressions into the equation
for the population rate:

\[
\frac{dN}{dt}=-\gamma\int d\boldsymbol{x}\left(g^{(3)}n^{3}-\frac{9}{2}\tilde{\delta}^{2}n-\frac{3}{2}\tilde{\delta}^{3}\right).
\]
We see that the second highest term in the expression is canceled,
which agrees with the expansion being correct up to the order $1/N$.
If the quantum correction term to the drift is omitted, one finds
that a physically incorrect quadratic nonlinear term proportional
to $n^{2}$ is obtained, which is inconsistent with an exact short-time
solution to the master equation\cite{Norrie2006a}.

\subsection{Two-component example}

As a more involved example, let us consider a two component $^{87}$Rb
BEC from recent experiments~\cite{Egorov2011,Opanchuk2012}. In this
case we have both unitary evolution (including nonlinear interaction)~\eqref{eqn:master-eqn:effective-H},
and three sources of losses: three-body recombination $\hat{O}_{111}=\tilde{\Psi}_{1}^{3}$,
two-body interspecies loss $\hat{O}_{12}=\tilde{\Psi}_{1}\tilde{\Psi}_{2}$
and two-body intraspecies loss $\hat{O}_{22}=\tilde{\Psi}_{2}^{2}$.
This gives us SDEs~\eqref{eqn:fpe:sdes} with drift terms

\begin{eqnarray*}
\mathcal{A}_{1} & = & -\frac{i}{\hbar}\left(\sum_{k=1}^{2}K_{1k}\Psi_{k}+\Psi_{1}\sum_{k=1}^{2}U_{1k}\left(\vert\Psi_{k}\vert^{2}-\frac{\delta_{1k}+1}{2}\tilde{\delta}_{k}\right)\right)\\
 &  & -3\kappa_{111}\left(\vert\Psi_{1}\vert^{2}-3\tilde{\delta}_{1}\right)\vert\Psi_{1}\vert^{2}\Psi_{1}-\kappa_{12}\left(\vert\Psi_{2}\vert^{2}-\frac{\tilde{\delta}_{2}}{2}\right)\Psi_{1},
\end{eqnarray*}

\begin{eqnarray*}
\mathcal{A}_{2} & = & -\frac{i}{\hbar}\left(\sum_{k=1}^{2}K_{2k}\Psi_{k}+\Psi_{2}\sum_{k=1}^{C}U_{2k}\left(\vert\Psi_{k}\vert^{2}-\frac{\delta_{2k}+1}{2}\tilde{\delta}_{k}\right)\right)\\
 &  & -\kappa_{12}\left(\vert\Psi_{1}\vert^{2}-\frac{\tilde{\delta}_{1}}{2}\right)\Psi_{2}-2\kappa_{22}\left(\vert\Psi_{2}\vert^{2}-\tilde{\delta}_{2}\right)\Psi_{2}.
\end{eqnarray*}
and noise terms

\[
\mathcal{B}_{1,111}=3\sqrt{\kappa_{111}}\left(\Psi_{1}^{*}\right)^{2},\quad\mathcal{B}_{1,12}=\sqrt{\kappa_{12}}\Psi_{2}^{*},\quad\mathcal{B}_{1,22}=0,
\]

\[
\mathcal{B}_{2,111}=0,\quad\mathcal{B}_{2,12}=\sqrt{\kappa_{12}}\Psi_{1}^{*},\quad\mathcal{B}_{2,22}=2\sqrt{\kappa_{22}}\Psi_{2}^{*}.
\]

This type of stochastic equation is needed to treat coherent BEC interferometry
in the presence of nonlinear loss terms caused by two and three body
collisions.

\subsection{Initial states}

Initial values for the numerical integration of equations~\eqref{eqn:fpe:sdes}
are obtained by finding the Wigner transformation of the density matrix
for the desired initial state, and then sampling the initial values
according to the resulting Wigner function. As an example of the procedure,
consider the simple case with a single-component coherent initial
state.
\begin{thm}
The Wigner distribution for a multi-mode coherent state with the expectation
value $\Psi^{(0)}\equiv\sum_{\boldsymbol{n}\in\mathbb{M}}\alpha_{\boldsymbol{n}}^{(0)}\phi_{\boldsymbol{n}}$
is
\begin{equation}
W_{c}[\Psi,\Psi^{*}]=\left(\frac{2}{\pi}\right)^{|\mathbb{M}|}\prod_{\boldsymbol{n}\in\mathbb{M}}\exp(-2|\alpha_{\boldsymbol{n}}-\alpha_{\boldsymbol{n}}^{(0)}|^{2}),
\end{equation}

where $\Psi\equiv\sum_{\boldsymbol{n}\in\mathbb{M}}\alpha_{\boldsymbol{n}}\phi_{\boldsymbol{n}}$.\end{thm}
\begin{proof}
The density matrix of the state is
\begin{equation}
\hat{\rho}=\vert\alpha_{\boldsymbol{n}}^{(0)},\,\boldsymbol{n}\in\mathbb{M}\rangle\langle\alpha_{\boldsymbol{n}}^{(0)},\,\boldsymbol{n}\in\mathbb{M}\vert=\left(\prod_{\boldsymbol{n}\in\mathbb{M}}\vert\alpha_{\boldsymbol{n}}^{(0)}\rangle\right)\left(\prod_{\boldsymbol{n}\in\mathbb{M}}\langle\alpha_{\boldsymbol{n}}^{(0)}\vert\right).
\end{equation}
Then the characteristic functional for this state can be expressed
as
\begin{equation}
\chi_{W}[\Lambda,\Lambda^{*}]=\prod_{\boldsymbol{n}\in\mathbb{M}}\langle\alpha_{\boldsymbol{n}}^{(0)}\vert\hat{D}_{\boldsymbol{n}}(\lambda_{\boldsymbol{n}},\lambda_{\boldsymbol{n}}^{*})\vert\alpha_{\boldsymbol{n}}^{(0)}\rangle,
\end{equation}
where $\lambda_{\boldsymbol{n}}$ are coefficients in the decomposition
of $\Lambda\in\mathbb{F}_{\mathbb{M}}$. Using the properties of the
displacement operator, this can be transformed to
\begin{equation}
\chi_{W}[\Lambda,\Lambda^{*}]=\prod_{\boldsymbol{n}\in\mathbb{M}}\exp(-\lambda_{\boldsymbol{n}}^{*}\alpha_{\boldsymbol{n}}^{(0)}+\lambda_{\boldsymbol{n}}(\alpha_{\boldsymbol{n}}^{(0)})^{*}-\frac{1}{2}|\lambda|^{2}).
\end{equation}
Finally, the Wigner function is
\begin{eqnarray*}
W_{c}[\Psi,\Psi^{*}] & = & \frac{1}{\pi^{2|\mathbb{M}|}}\prod_{\boldsymbol{n}\in\mathbb{M}}\left(\int d^{2}\lambda_{\boldsymbol{n}}\exp(-\lambda_{\boldsymbol{n}}(\alpha_{\boldsymbol{n}}^{*}-(\alpha_{\boldsymbol{n}}^{(0)})^{*})+\lambda_{\boldsymbol{n}}^{*}(\alpha_{\boldsymbol{n}}-\alpha_{\boldsymbol{n}}^{(0)})-\frac{1}{2}|\lambda|^{2})\right)\\
 & = & \left(\frac{2}{\pi}\right)^{|\mathbb{M}|}\prod_{\boldsymbol{n}\in\mathbb{M}}\exp(-2|\alpha_{\boldsymbol{n}}-\alpha_{\boldsymbol{n}}^{(0)}|^{2}).
\end{eqnarray*}

\end{proof}
The resulting Wigner distribution is a product of independent complex-valued
Gaussian distributions for each mode, with an expectation value equal
to the expectation value of the mode, and variance equal to $\frac{1}{2}$.
Therefore the initial state can be sampled as
\begin{equation}
\alpha_{\boldsymbol{n}}=\alpha_{\boldsymbol{n}}^{(0)}+\frac{1}{\sqrt{2}}\eta_{\boldsymbol{n}},
\end{equation}
where $\eta_{\boldsymbol{n}}$ are normally distributed complex random
numbers with zero mean, $\langle\eta_{\boldsymbol{m}}\eta_{\boldsymbol{n}}\rangle=0$
and $\langle\eta_{\boldsymbol{m}}\eta_{\boldsymbol{n}}^{*}\rangle=\delta_{\boldsymbol{m},\boldsymbol{n}}$
or, in other words, with real components distributed independently
with variance $\frac{1}{2}$. This looks like adding half a ``vacuum
particle'' to each mode. In functional form this can be written as
\[
\Psi(\boldsymbol{x},0)=\Psi^{(0)}(\boldsymbol{x},0)+\sum_{\boldsymbol{n}\in\mathbb{M}}\frac{\eta_{\boldsymbol{n}}}{\sqrt{2}}\phi_{\boldsymbol{n}},
\]
where $\Psi^{(0)}(\boldsymbol{x},0)$ is the ``classical'' ground
state of the system.

More involved examples, including thermalized states and Bogoliubov
states, are reviewed by Blakie et al.~\cite{Blakie2008}, and Ruostekoski
and Martin~\cite{Ruostekoski2010}. In particular, a numerically
efficient way to sample a Wigner distribution for Bogoliubov states
was developed by Sinatra et al.~\cite{Sinatra2002}

\section{Conclusion}

We have formally derived all the equations necessary to describe BEC
interferometry experiments statistically, given a master equation
written in terms of field operators. We have provided general equations
required to use the transformation, along with its application to
the trapped BEC case. In the latter case, the resulting SDEs can be
integrated numerically using conventional methods, and their solutions
can be used to calculate all the required observables.

\appendix
\counterwithin{thm}{section} 

\section{Wirtinger differentiation}

In this paper we are using differentiation of complex functions extensively.
Instead of the classical definition of a differential which only works
for holomorphic functions, we use Wirtinger differentiation~\cite{Wirtinger1927}.
One can find thorough descriptions of these rules, for example, in~\cite{Kreutz-Delgado2009};
in this section we will only outline the basics.
\begin{defn}
For a complex variable $z=x+iy$ and a function $f(z)=u(x,y)+iv(x,y)$
the Wirtinger differential is
\[
\frac{\partial f(z)}{\partial z}=\frac{1}{2}\left(\frac{\partial f}{\partial x}-i\frac{\partial f}{\partial y}\right).
\]

\end{defn}
One can easily prove that if $f(z)$ is holomorphic, then the above
definition coincides with the classical differential for complex functions.
Wirtinger differential obeys sum, product, quotient, and chain differentiation
rules (the former is applied as if $f(z)\equiv f(z,z^{*})$).

In addition, we will need an area integration over a complex variable:
\begin{defn}
For a complex variable $z=x+iy$ the integral
\[
\int d^{2}z\equiv\int_{-\infty}^{\infty}\int_{-\infty}^{\infty}dx\, dy,
\]
or, in other words, this stands for a two-dimensional integral over
the complex plane.
\end{defn}
Such integration has a property similar to a Fourier transformation
in real space.
\begin{lem}
\label{lmm:c-numbers:fourier-of-moments} If $\lambda$ is a complex
variable, then for any non-negative integers $r$ and $s$:
\[
\int d^{2}\alpha\,\alpha^{r}(\alpha^{*})^{s}\exp(-\lambda\alpha^{*}+\lambda^{*}\alpha)=\pi^{2}\left(-\frac{\partial}{\partial\lambda^{*}}\right)^{r}\left(\frac{\partial}{\partial\lambda}\right)^{s}\delta(\mathrm{Re}\lambda)\delta(\mathrm{Im}\lambda)
\]
\end{lem}
\begin{proof}
First, using known Fourier transform relations, it is easy to prove
that for real $x$ and $v$, and non-negative integer $n$
\[
\int\limits _{-\infty}^{\infty}dv\, v^{n}\exp(\pm2ixv)=\pi(\mp i/2)^{n}\delta^{(n)}(x).
\]
Substituting $\alpha=x+iy$, expanding the $\alpha^{r}(\alpha^{*})^{s}$
term using binomial theorem and employing the above property, one
can reach the statement of the lemma.
\end{proof}
Another important property is used extensively throughout the paper.
\begin{lem}
\label{lmm:c-numbers:zero-integrals} If $f(\lambda,\lambda^{*})$
is square-integrable, then for any complex $\alpha$:
\begin{eqnarray*}
\int d^{2}\lambda\frac{\partial}{\partial\lambda}\left(\exp(-\lambda\alpha^{*}+\lambda^{*}\alpha)f(\lambda,\lambda^{*})\right) & =0,\\
\int d^{2}\lambda\frac{\partial}{\partial\lambda^{*}}\left(\exp(-\lambda\alpha^{*}+\lambda^{*}\alpha)f(\lambda,\lambda^{*})\right) & =0.
\end{eqnarray*}
 \end{lem}
\begin{proof}
Square-integrability of f means $\underset{\mathrm{Re}\lambda\rightarrow\infty}{\lim}f=0$
and $\underset{\mathrm{Im}\lambda\rightarrow\infty}{\lim}f=0$, so
the statement of the lemma can be proved by transforming to real variables
and integrating.
\end{proof}

\section{Functional calculus\label{sec:app-functional-calculus}}

This section outlines the functional calculus, which is heavily used
throughout the paper. A detailed description is given in~\cite{Dalton2011},
and here we only provide some important definitions and results which
are used later in the paper. In this section we will use the definitions
from Section~\ref{sec:func-operators}, namely the full basis $\mathbb{B}$
and the restricted basis $\mathbb{M}$. Given the basis, we can define
a correspondence between functions of coordinates and their representations
in mode space.
\begin{defn}
Let $\mathbb{F}$ be the space of all functions of coordinates, which
consists only of modes from $\mathbb{M}$: $\mathbb{F}_{\mathbb{M}}\equiv(\mathbb{R}^{D}\rightarrow\mathbb{C})_{\mathbb{M}}$
(restricted functions). The composition transformation $\mathcal{C}_{\mathbb{M}}\in\mathbb{C}^{|\mathbb{M}|}\rightarrow\mathbb{F}_{\mathbb{M}}$
creates a function from a vector of mode populations:
\[
\mathcal{C}_{\mathbb{M}}(\boldsymbol{\alpha})=\sum_{\boldsymbol{n}\in\mathbb{M}}\phi_{\boldsymbol{n}}\alpha_{\boldsymbol{n}}.
\]
The decomposition transformation $\mathcal{C}_{\mathbb{M}}^{-1}\in\mathbb{F}\rightarrow\mathbb{C}^{|\mathbb{M}|}$,
correspondingly, creates a vector of populations out of a function:
\[
(\mathcal{C}_{\mathbb{M}}^{-1}[f])_{\boldsymbol{n}}=\int d\boldsymbol{x}\phi_{\boldsymbol{n}}^{*}f,\,\boldsymbol{n}\in\mathbb{M}.
\]
Note that for any $f\in\mathbb{F}_{\mathbb{M}}$, $\mathcal{C}_{\mathbb{M}}(\mathcal{C}_{\mathbb{M}}^{-1}[f])\equiv f$.
\end{defn}
The result of any non-linear transformation of a function $f\in\mathbb{F}_{\mathbb{M}}$
is not guaranteed to belong to $\mathbb{F}_{\mathbb{M}}$. This requires
explicit projections to be used with other restricted functions. This
also applies to the delta function of coordinates. To avoid confusion
with the common delta function, we introduce the restricted delta
function.
\begin{defn}
\label{def:func-calculus:restricted-delta} The restricted delta function
$\delta_{\mathbb{M}}\in\mathbb{F}_{\mathbb{M}}$ is defined as
\[
\delta_{\mathbb{M}}(\boldsymbol{x}^{\prime},\boldsymbol{x})=\sum_{\boldsymbol{n}\in\mathbb{M}}\phi_{\boldsymbol{n}}^{\prime*}\phi_{\boldsymbol{n}}.
\]
Note that $\delta_{\mathbb{M}}^{*}(\boldsymbol{x}^{\prime},\boldsymbol{x})=\delta_{\mathbb{M}}(\boldsymbol{x},\boldsymbol{x}^{\prime})$.
\end{defn}
Any function can be projected to $\mathbb{M}$ using the projection
transformation.
\begin{defn}
\label{def:func-calculus:projector} Projection transformation $\mathcal{P}_{\mathbb{M}}\in\mathbb{F}\rightarrow\mathbb{F}_{\mathbb{M}}$
\[
\mathcal{P}_{\mathbb{M}}[f](\boldsymbol{x})=(\mathcal{C}_{\mathbb{M}}(\mathcal{C}_{\mathbb{M}}^{-1}[f]))(\boldsymbol{x})=\sum_{\boldsymbol{n}\in\mathbb{M}}\phi_{\boldsymbol{n}}\int d\boldsymbol{x}^{\prime}\,\phi_{\boldsymbol{n}}^{\prime*}f^{\prime}=\int d\boldsymbol{x}^{\prime}\delta_{\mathbb{M}}(\boldsymbol{x}^{\prime},\boldsymbol{x})f^{\prime}.
\]
Obviously, $\mathcal{P}_{\mathbb{B}}\equiv\mathds{1}$. The conjugate
of $\mathcal{P}_{\mathbb{M}}$ is thus defined as
\[
(\mathcal{P}_{\mathbb{M}}[f](\boldsymbol{x}))^{*}=\int d\boldsymbol{x}^{\prime}\delta_{\mathbb{M}}^{*}(\boldsymbol{x}^{\prime},\boldsymbol{x})f^{\prime*}=\mathcal{P}_{\mathbb{M}}^{*}[f^{*}](\boldsymbol{x}).
\]

\end{defn}
Let $\mathcal{F}[f]\in\mathbb{F}_{\mathbb{M}}\rightarrow\mathbb{F}$
be some transformation (note that the result is not guaranteed to
belong to the restricted basis). Because of the bijection between
$\mathbb{F}_{\mathbb{M}}$ and $\mathbb{C}^{|\mathbb{M}|}$, $\mathcal{F}$
can be alternatively treated as a function of a vector of complex
numbers $\mathcal{F}\in\mathbb{C}^{|\mathbb{M}|}\rightarrow\mathbb{C}^{\infty}$:
\[
\mathcal{F}(\boldsymbol{\alpha})\equiv\mathcal{C}_{\mathbb{M}}^{-1}[\mathcal{F}[\mathcal{C}_{\mathbb{M}}(\boldsymbol{\alpha})]].
\]
Using this correspondence, we can define functional differentiation.
\begin{defn}
\label{def:func-calculus:func-diff} The functional derivative $\frac{\delta}{\delta f^{\prime}}\in\left(\mathbb{F}_{\mathbb{M}}\rightarrow\mathbb{F}\right)\rightarrow\left(\mathbb{R}^{D}\rightarrow\mathbb{F}_{\mathbb{M}}\rightarrow\mathbb{F}\right)$
is defined as
\[
\frac{\delta\mathcal{F}[f]}{\delta f^{\prime}}=\sum_{\boldsymbol{n}\in\mathbb{M}}\phi_{\boldsymbol{n}}^{\prime*}\frac{\partial\mathcal{F}(\boldsymbol{\alpha})}{\partial\alpha_{\boldsymbol{n}}}.
\]

\end{defn}
Note that the transformation being returned differs from the one which
was taken: the result of the new transformation is a function of the
additional variable from $\mathbb{R}^{D}$ ($\boldsymbol{x}^{\prime}$).
This variable comes from the function we are differentiating by.

Functional derivatives behave in many ways similar to Wirtinger derivatives.
A detailed treatment can be found in~\cite{Dalton2011}. In particular,
the following useful lemma gives us the ability to differentiate functionals
in a similar way to common functions:
\begin{lem}
If $g(z)$ is a function of complex variable that can be expanded
into series of $z^{n}(z^{*})^{m}$, and functional $\mathcal{F}[f,f^{*}]\equiv g(f,f^{*})$,
$\mathcal{F}\in\mathbb{F}_{\mathbb{M}}\rightarrow\mathbb{F}$, then
$\delta\mathcal{F}/\delta f^{\prime}$ and $\delta\mathcal{F}/\delta f^{\prime*}$
can be treated as partial differentiation of the functional of two
independent variables $f$ and $f^{*}$. In other words:
\[
\frac{\delta\mathcal{F}}{\delta f^{\prime}}=\delta_{\mathbb{M}}(\boldsymbol{x}^{\prime},\boldsymbol{x})\frac{\partial g(f,f^{*})}{\partial f},\qquad\frac{\delta\mathcal{F}}{\delta f^{\prime*}}=\delta_{\mathbb{M}}^{*}(\boldsymbol{x}^{\prime},\boldsymbol{x})\frac{\partial g(f,f^{*})}{\partial f^{*}}
\]
\end{lem}
\begin{defn}
\label{def:func-calculus:func-integration} Functional integration
$\int\delta^{2}f\in(\mathbb{F}_{\mathbb{M}}\rightarrow\mathbb{F})\rightarrow\mathbb{C}$
is defined as
\[
\int\delta^{2}f\mathcal{F}[f]=\int d^{2}\boldsymbol{\alpha}\mathcal{F}(\boldsymbol{\alpha})
\]
If the basis contains an infinite number of modes, the integral is
treated as a limit $|\mathbb{M}|\rightarrow\infty$.
\end{defn}
Functional integration has the Fourier-like property analogous to
Lemma~\ref{lmm:c-numbers:fourier-of-moments}, but its statement
requires the definition of the delta functional:
\begin{defn}
\label{def:func-calculus:delta-functional} For a function $\Lambda\in\mathbb{F}_{\mathbb{M}}$
the delta functional is
\[
\Delta_{\mathbb{M}}[\Lambda]\equiv\prod_{\boldsymbol{n}\in\mathbb{M}}\delta(\mathrm{Re}\lambda_{\boldsymbol{n}})\delta(\mathrm{Im}\lambda_{\boldsymbol{n}}),
\]
where $\boldsymbol{\lambda}=\mathcal{C}_{\mathbb{M}}^{-1}[\Lambda]$.
\end{defn}
The delta functional has the same property as the common delta function:
\begin{eqnarray}
\int\delta^{2}\Lambda\mathcal{F}[\Lambda]\Delta_{\mathbb{M}}[\Lambda] & = & \int d^{2}\boldsymbol{\lambda}\mathcal{F}(\boldsymbol{\lambda})\prod_{\boldsymbol{n}\in\mathbb{M}}\delta(\mathrm{Re}\lambda_{\boldsymbol{n}})\delta(\mathrm{Im}\lambda_{\boldsymbol{n}})\nonumber \\
 & = & \left.\mathcal{F}(\boldsymbol{\lambda})\right|_{\forall\boldsymbol{n}\in\mathbb{M}\,\lambda_{\boldsymbol{n}}=0}\nonumber \\
 & = & \left.\mathcal{F}[\Lambda]\right|_{\Lambda\equiv0}
\end{eqnarray}

\begin{lem}[Functional extension of Lemma~\ref{lmm:c-numbers:fourier-of-moments}]
\label{lmm:func-calculus:fourier-of-moments} For $\Psi\in\mathbb{F}_{\mathbb{M}}$
and $\Lambda\in\mathbb{F}_{\mathbb{M}}$, and for any non-negative
integers $r$ and $s$:
\begin{eqnarray*}
 &  & \int\delta^{2}\Psi\,\Psi^{r}(\Psi^{*})^{s}\exp\left(\int d\boldsymbol{x}\left(-\Lambda\Psi^{*}+\Lambda^{*}\Psi\right)\right)\\
 &  & =\pi^{2|\mathbb{M}|}\left(-\frac{\delta}{\delta\Lambda^{*}}\right)^{r}\left(\frac{\delta}{\delta\Lambda}\right)^{s}\Delta_{\mathbb{M}}[\Lambda]
\end{eqnarray*}
\end{lem}
\begin{proof}
The proof consists of expanding functions into sums of modes and applying
Lemma~\ref{lmm:c-numbers:fourier-of-moments} $|\mathbb{M}|$ times. \end{proof}
\begin{lem}[Functional extension of Lemma~\ref{lmm:c-numbers:zero-integrals}]
\label{lmm:func-calculus:zero-integrals} For a square-integrable
functional $F$
\begin{eqnarray*}
\int\delta^{2}\Lambda\frac{\delta}{\delta\Lambda^{\prime}}\left(D[\Lambda,\Lambda^{*},\Psi,\Psi^{*}]F[\Lambda,\Lambda^{*}]\right) & =0\\
\int\delta^{2}\Lambda\frac{\delta}{\delta\Lambda^{\prime*}}\left(D[\Lambda,\Lambda^{*},\Psi,\Psi^{*}]F[\Lambda,\Lambda^{*}]\right) & =0.
\end{eqnarray*}
\end{lem}
\begin{proof}
Proved by expanding integrals and differentials into modes and applying
Lemma~\ref{lmm:c-numbers:zero-integrals}. \end{proof}
\begin{lem}
\label{lmm:func-calculus:zero-delta-integrals} For $\Lambda\in\mathbb{F}_{\mathbb{M}}$
and a bounded functional $F$
\begin{eqnarray*}
\int\delta^{2}\Lambda\frac{\delta}{\delta\Lambda}\left(\left(\left(\frac{\delta}{\delta\Lambda}\right)^{s}\left(-\frac{\delta}{\delta\Lambda^{*}}\right)^{r}\Delta_{\mathbb{M}}[\Lambda]\right)F[\Lambda,\Lambda^{*}]\right) & =0\\
\int\delta^{2}\Lambda\frac{\delta}{\delta\Lambda^{*}}\left(\left(\left(\frac{\delta}{\delta\Lambda}\right)^{s}\left(-\frac{\delta}{\delta\Lambda^{*}}\right)^{r}\Delta_{\mathbb{M}}[\Lambda]\right)F[\Lambda,\Lambda^{*}]\right) & =0
\end{eqnarray*}
 \end{lem}
\begin{proof}
Proved by expanding functional integration and differentials into
modes and integrating separately over each $\lambda_{\boldsymbol{n}}$,
using the fact that any differential of the delta function is zero
on the infinity.
\end{proof}
In order to perform transformations of master equations, we will need
a lemma that justifies the ``relocation'' of the Laplacian (which
is a part of the kinetic term in the Hamiltonian) inside the functional
integral.
\begin{lem}
\label{lmm:func-calculus:move-laplacian} If $\mathcal{F}\in\mathbb{F}_{\mathbb{M}}\rightarrow\mathbb{F}$,
and $\forall\boldsymbol{n}\in\mathbb{M},\boldsymbol{x}\in\partial A$
$\phi_{\boldsymbol{n}}(\boldsymbol{x})=0$, then
\[
\int\limits _{A}d\boldsymbol{x}\left(\nabla^{2}\frac{\delta}{\delta\Psi}\right)\Psi\mathcal{F}[\Psi,\Psi^{*}]=\int\limits _{A}d\boldsymbol{x}\frac{\delta}{\delta\Psi}(\nabla^{2}\Psi)\mathcal{F}[\Psi,\Psi^{*}]
\]
\end{lem}
\begin{proof}
The proof consists of a function expansion into a mode sum and an
application of Green's first identity.
\end{proof}
Note that the above lemma imposes an additional requirement for basis
functions, but in practical applications it is always satisfied. For
example, in a plane wave basis eigenfunctions are equal to zero at
the border of the bounding box, and in a harmonic oscillator basis
they are equal to zero on the infinity (which can be considered the
boundary of their integration area). We will assume that this condition
is true for any basis we work with.

\section{Functional Fokker-Planck equation\label{sec:app-functional-fpe}}

The general approach to numerical solution of the Fokker-Planck equation
is to transform it to the equivalent set of stochastic differential
equations (SDEs). In the textbooks this transformation is defined
for real variables only~\cite{Risken1996}, while we have functional
FPE with complex-valued functions.

Our starting point is the reformulation of the theorem for real-valued
multivariable FPE from~\cite{Risken1996} in terms of vectors and
matrices:
\begin{lem}[FPE\textendash{}SDEs correspondence in convenient form]
\label{lmm:app-fpe:fpe-sde-real}If $\boldsymbol{z}^{T}\equiv(z_{1}\ldots z_{M})$
is a set of real-valued variables, Fokker-Planck equation
\[
\frac{dW}{dt}=-\boldsymbol{\partial}_{\boldsymbol{z}}^{T}\boldsymbol{a}W+\frac{1}{2}\mathrm{Tr}\left\{ \boldsymbol{\partial}_{\boldsymbol{z}}\boldsymbol{\partial}_{\boldsymbol{z}}^{T}BB^{T}\right\} W
\]
is equivalent to a set of stochastic differential equations in Itô
form
\[
d\boldsymbol{z}=\boldsymbol{a}dt+Bd\boldsymbol{Z}
\]
and to a set of stochastic differential equations in Stratonovich
form
\[
d\boldsymbol{z}=(\boldsymbol{a}-\boldsymbol{s})dt+Bd\boldsymbol{Z},
\]
where the noise-induced (or spurious) drift vector $\boldsymbol{s}$
has elements
\[
s_{i}=\frac{1}{2}\sum_{k,j}B_{kj}\frac{\partial}{\partial z_{k}}B_{ij}=\frac{1}{2}\mathrm{Tr}\left\{ B^{T}\boldsymbol{\partial}_{z}\boldsymbol{e}_{i}^{T}B\right\} ,
\]
$\boldsymbol{e}_{i}$ being the unit vector with elements $(e_{i})_{j}=\delta_{ij}$.
Here $W\equiv W(\boldsymbol{z})$ is a probability distribution, $\boldsymbol{a}\equiv\boldsymbol{a}(\boldsymbol{z})$
is a vector function, $B\equiv B(\boldsymbol{z})$ is a matrix function
($B$ having size $M\times L$, where $L$ corresponds to the number
of noise sources), $\boldsymbol{\partial}_{\boldsymbol{z}}^{T}\equiv(\partial_{z_{1}}\ldots\partial_{z_{M}})$
is a vector differential, and $\boldsymbol{Z}$ is a standard $L$-dimensional
real-valued Wiener process. \end{lem}
\begin{proof}
For details see~\cite{Risken1996}, sections 3.3 and 3.4. \end{proof}
\begin{thm}
\label{thm:app-fpe:fpe-sde-complex} If $\boldsymbol{\alpha}^{T}\equiv(\alpha_{1}\ldots\alpha_{M})$
is a set of complex-valued variables, Fokker-Planck equation
\[
\frac{dW}{dt}=-\boldsymbol{\partial}_{\boldsymbol{\alpha}}^{T}\boldsymbol{a}W-\boldsymbol{\partial}_{\boldsymbol{\alpha}^{*}}^{T}\boldsymbol{a}^{*}W+\mathrm{Tr}\left\{ \boldsymbol{\partial}_{\boldsymbol{\alpha}^{*}}\boldsymbol{\partial}_{\boldsymbol{\alpha}}^{T}BB^{H}\right\} W
\]
is equivalent to a set of stochastic differential equations in Itô
form
\[
d\boldsymbol{\alpha}=\boldsymbol{a}dt+Bd\boldsymbol{Z},
\]
or to Stratonovich form
\[
d\boldsymbol{\alpha}=(\boldsymbol{a}-\boldsymbol{s})dt+Bd\boldsymbol{Z},
\]
where noise-induced drift term is
\[
s_{j}=\frac{1}{2}\mathrm{Tr}\left\{ B^{H}\boldsymbol{\partial}_{\boldsymbol{\alpha}^{*}}\boldsymbol{e}_{j}^{T}B\right\} ,
\]
and $\boldsymbol{Z}=(\boldsymbol{X}+i\boldsymbol{Y})/\sqrt{2}$ is
a standard $L$-dimensional complex-valued Wiener process, containing
two standard real-valued $L$-dimensional Wiener processes $\boldsymbol{X}$
and $\boldsymbol{Y}$. \end{thm}
\begin{proof}
Proved straightforwardly by transforming the equation to real variables
and applying Lemma~\ref{lmm:app-fpe:fpe-sde-real}. \end{proof}
\begin{thm}
\label{thm:app-fpe:mc-fpe-sde} If $\boldsymbol{\alpha}^{(j)},\, j=1..C$
are $C$ sets of complex variables $\boldsymbol{\alpha}^{(j)}\equiv(\alpha_{1}^{(j)}\ldots\alpha_{M_{j}}^{(j)})$,
then the Fokker-Planck equation
\begin{eqnarray*}
\frac{dW}{dt}= &  & -\sum_{j=1}^{C}\boldsymbol{\partial}_{\boldsymbol{\alpha}^{(j)}}^{T}\boldsymbol{a}^{(j)}W-\sum_{j=1}^{C}\boldsymbol{\partial}_{(\boldsymbol{\alpha}^{(j)})^{*}}^{T}(\boldsymbol{a}^{(j)})^{*}W\\
 &  & +\sum_{j=1}^{C}\sum_{k=1}^{C}\mathrm{Tr}\left\{ \boldsymbol{\partial}_{(\boldsymbol{\alpha}^{(j)})^{*}}\boldsymbol{\partial}_{\boldsymbol{\alpha}^{(k)}}^{T}B^{(k)}(B^{(j)})^{H}\right\} W
\end{eqnarray*}
is equivalent to a set of stochastic differential equations in Itô
form
\[
d\boldsymbol{\alpha}^{(j)}=\boldsymbol{a}^{(j)}dt+B^{(j)}d\boldsymbol{Z},\, j=1..C
\]
or to Stratonovich form
\[
d\boldsymbol{\alpha}^{(j)}=(\boldsymbol{a}^{(j)}-\boldsymbol{s}^{(j)})dt+B^{(j)}d\boldsymbol{Z},
\]
where noise-induced drift term is
\[
s_{i}^{(j)}=\frac{1}{2}\sum_{k=1}^{C}\mathrm{Tr}\left\{ (B^{(k)})^{H}\boldsymbol{\partial}_{(\boldsymbol{\alpha}^{(k)})^{*}}\boldsymbol{e}_{i}^{T}B^{(j)}\right\} ,
\]
and $\boldsymbol{Z}$ is a standard $L$-dimensional complex-valued
Wiener process. \end{thm}
\begin{proof}
Proved by joining vectors from all components into one vector and
applying Theorem~\ref{thm:app-fpe:fpe-sde-complex}. \end{proof}
\begin{thm}
\label{thm:app-fpe:fpe-sde-func} For a probability distribution $W[\boldsymbol{\Psi},\boldsymbol{\Psi}^{*}]\in\mathbb{F}_{\mathbb{M}}^{C}\rightarrow\mathbb{R}$,
a $C$-dimensional vector of transformations $\boldsymbol{\mathcal{A}}$
and a $C\times L$ matrix of transformations $\mathcal{B}$ the functional
FPE
\[
\frac{dW}{dt}=\int d\boldsymbol{x}\left(-2\mathrm{Re}\left(\boldsymbol{\delta}_{\boldsymbol{\Psi}}\cdot\boldsymbol{\mathcal{A}}\right)+\mathrm{Tr}\left\{ \boldsymbol{\delta}_{\boldsymbol{\Psi}^{*}}\boldsymbol{\delta}_{\boldsymbol{\Psi}}^{T}\mathcal{B}\mathcal{B}^{H}\right\} \right)W
\]
is equivalent to the set of SDEs in Itô form
\[
d\boldsymbol{\Psi}=\boldsymbol{\mathcal{P}}\left[\boldsymbol{\mathcal{A}}dt+\mathcal{B}d\boldsymbol{Q}\right]
\]
or in Stratonovich form
\[
d\boldsymbol{\Psi}=\boldsymbol{\mathcal{P}}\left[(\boldsymbol{\mathcal{A}}-\boldsymbol{\mathcal{S}})dt+\mathcal{B}d\boldsymbol{Q}\right],
\]
where
\[
\mathcal{S}_{j}=\frac{1}{2}\mathrm{Tr}\left\{ \mathcal{B}^{H}\boldsymbol{\delta}_{\boldsymbol{\Psi}^{*}}\boldsymbol{e}_{j}^{T}\mathcal{B}\right\} ,
\]
$\boldsymbol{Q}$ is an $L$-dimensional vector of standard functional
Wiener processes:
\[
Q_{l}=\sum_{\boldsymbol{n}\in\mathbb{B}}\phi_{\boldsymbol{n}}Z_{l,\boldsymbol{n}}
\]

and $\boldsymbol{\mathcal{P}}^{T}=\left(\mathcal{P}_{\mathbb{M}_{1}},\ldots,\mathcal{P}_{\mathbb{M}_{C}}\right)$
is a vector of projection transformations.\end{thm}
\begin{proof}
Proved by expanding functional derivatives and applying Theorem~\ref{thm:app-fpe:mc-fpe-sde}.
The diffusion term has to be transformed in order to conform to the
theorem:
\begin{eqnarray}
\int d\boldsymbol{x}\phi_{j,\boldsymbol{m}}\phi_{k,\boldsymbol{n}}^{*}\sum_{l=1}^{L}\mathcal{B}_{kl}\mathcal{B}_{jl}^{*} & = & \int d\boldsymbol{x}\int d\boldsymbol{x}^{\prime}\phi_{j,\boldsymbol{m}}^{\prime}\phi_{k,\boldsymbol{n}}^{*}\sum_{l=1}^{L}\mathcal{B}_{jl}^{\prime*}\mathcal{B}_{kl}\delta(\boldsymbol{x}-\boldsymbol{x}^{\prime})\nonumber \\
 & = & \int d\boldsymbol{x}\int d\boldsymbol{x}^{\prime}\phi_{j,\boldsymbol{m}}^{\prime}\phi_{k,\boldsymbol{n}}^{*}\sum_{l=1}^{L}\mathcal{B}_{jl}^{\prime*}\mathcal{B}_{kl}\sum_{\boldsymbol{p}\in\mathbb{B}}\phi_{\boldsymbol{p}}^{\prime*}\phi_{\boldsymbol{p}}\nonumber \\
 & = & \sum_{l=1}^{L}\sum_{\boldsymbol{p}\in\mathbb{B}}\int d\boldsymbol{x}\phi_{j,\boldsymbol{m}}\mathcal{B}_{jl}^{*}\phi_{\boldsymbol{p}}^{*}\int d\boldsymbol{x}\phi_{k,\boldsymbol{n}}^{*}\mathcal{B}_{kl}\phi_{\boldsymbol{p}}.
\end{eqnarray}
Grouping terms back and recognizing the definition of projection transformation,
one gets the statement of the theorem.
\end{proof}
\bibliography{qsim-long}

\end{document}